\documentclass[a4paper,10pt,openright]{article} 
\usepackage[english]{babel} 
\usepackage[latin1]{inputenc} 
\usepackage{amsmath,amsthm,amsfonts,amscd,eucal,latexsym,amssymb,mathrsfs,cancel,tikz}
\usepackage{epsfig}  
\usepackage{simplewick}
\usepackage{hyperref}
\usepackage{color}

\definecolor{NiColor}{RGB}{77,77,255}
\definecolor{NiColoRed}{RGB}{255,77,77}
\definecolor{NiCologReen}{RGB}{77,255,77}

\newtheoremstyle{TheoremStyle}
        {3pt}
        {3pt}
        {}
        {}
        {\bf}
        {:}
        {.5em}
        {}
\newtheoremstyle{ExampleStyle}
        {3pt}
        {3pt}
        {}
        {}
        {\bf}
        {:}
        {.5em}
        {}

\theoremstyle{TheoremStyle}
\newtheorem{theorem}{Theorem}
\newtheorem{corollary}[theorem]{Corollary}
\newtheorem{proposition}[theorem]{Proposition}
\newtheorem{lemma}[theorem]{Lemma}
\newtheorem{Definition}[theorem]{Definition}
\newtheorem{remark}[theorem]{Remark}
\theoremstyle{ExampleStyle}
\newtheorem{Example}[theorem]{Example} 

\title{Fundamental solutions for the wave operator on static Lorentzian manifolds with timelike boundary}
\author{Claudio Dappiaggi$^{1,2,a}$, Nicol\'o Drago$^{1,2,b}$, Hugo Ferreira$^{1,3,c}$ 
	\vspace{4mm}\\
	{\small $^1$ Istituto Nazionale di Alta Matematica}\\ 
	{\small Sezione di Pavia - Via Ferrata, 5, 27100 Pavia, Italy.}\vspace{2mm}\\
	{\small $^2$ Dipartimento di Fisica}\\ 
	{\small Universit{\`a} di Pavia - Via Bassi 6, 27100 Pavia, Italy.}\vspace{2mm}\\
	{\small $^3$ INFN, Sezione di Pavia - Via Bassi 6, 27100 Pavia, Italy.}\vspace{2mm}\\
	{\footnotesize  ~$^a$ claudio.dappiaggi@unipv.it~,~$^b$ nicolo.drago@unipv.it~,~$^c$ 
		hugo.ferreira@pv.infn.it }}

\begin{document}
\maketitle

\begin{abstract}
	We consider the wave operator on static, Lorentzian manifolds with timelike boundary and we discuss the existence of advanced and retarded fundamental solutions in terms of boundary conditions. By means of spectral calculus we prove that answering this question is equivalent to studying the self-adjoint extensions of an associated elliptic operator on a Riemannian manifold with boundary $(M,g)$. The latter is diffeomorphic to any, constant time hypersurface of the underlying background. In turn, assuming that $(M,g)$ is of bounded geometry, this problem can be tackled within the framework of boundary triples. These consist of the assignment of two surjective, trace operators from the domain of the adjoint of the elliptic operator onto an auxiliary Hilbert space $\mathsf{h}$, which is the third datum of the triple. Self-adjoint extensions of the underlying elliptic operator are in one-to-one correspondence with self-adjoint operators $\Theta$ on $\mathsf{h}$. On the one hand, we show that, for a natural choice of boundary triple, each $\Theta$ can be interpreted as the assignment of a boundary condition for the original wave operator. On the other hand, we prove that, for each such $\Theta$, there exists a unique advanced and retarded fundamental solution. In addition, we prove that these share the same structural property of the counterparts associated to the wave operator on a globally hyperbolic spacetime.
\end{abstract}

\section{Introduction}

The existence and the characterization of the fundamental solutions of the D'Alembert wave operator $\square$ on a Lorentzian manifold $(N,g)$ is a classical problem which has been thoroughly studied in many contexts. Particularly if the underlying background is globally hyperbolic, a complete answer is known, {\it cf.} \cite{BGP}, showing that there exist two unique distributions $\mathcal{G}^\pm\in\mathcal{D}^\prime(N\times N)$, called {\em advanced} ($-$) and {\em retarded} ($+$) fundamental solutions whose associated Green operators $\mathsf{G}^\pm:\mathcal{D}(N)\to C^\infty(N)$ are such that $\square\circ\mathsf{G}^\pm=\mathsf{G}^\pm\circ\square=\textrm{id}|_{\mathcal{D}(N)}$  and, for any $f\in\mathcal{D}(N)$, $\textrm{supp}(\mathsf{G}^\pm(f))\subseteq J^\mp(\textrm{supp}(f))$, $J^\pm$ being the causal future ($+$) and past ($-$).

Completely different is the situation if $(N,g)$ is a Lorentzian manifold with timelike boundary $(\partial N, \iota_N^*g)$, where $\iota_N:\partial N\hookrightarrow N$ and where $(\partial N, \iota_N^*g)$ is a Lorentzian submanifold. In this case a complete, cohesive analysis of the existence of fundamental solutions is not available. Yet, this class of backgrounds contains several notable examples, such as anti-de Sitter (AdS) or asymptotically AdS spacetimes which play a key r\^ole in several models that have recently attracted a lot of attention for the study of the properties of the wave or of the Klein-Gordon equation, see for example \cite{Bachelot,Holzegel,Wrochna:2016ruq,Vasy}. 

In addition, fundamental solutions play a pivotal r\^ole in the covariant quantization of free field theories, most notably in the construction of the $*$-algebra of observables. As a matter of facts, focusing on Bosonic fields, the canonical commutation relations are implemented in terms of a $*$-ideal which, in turn, is built out of a symplectic form defined out of the fundamental solutions, so to encode the information of dynamics and causality, see {\it e.g.} \cite{Benini:2013fia,Benini:2015bsa}. From this viewpoint, the question that we shall answer in this paper has a direct impact in our understanding of several physical systems, as existence of fundamental solutions guarantees on the one hand that any underlying quantum field theory can be constructed following the standard rationale. On the other hand, if these solutions turn out not to be unique, it will entail that each of the possible choices corresponds to a different physical quantum system, as it has been already observed in different contexts, {\it e.g.} \cite{Benini:2017dfw,Dappiaggi:2017wvj,Dappiaggi:2014gea}.

More precisely, our last statement stems from the observation that, in comparison to the same class of linear hyperbolic, partial differential equations on globally hyperbolic backgrounds, the main structural difference in our setting lies in the fact that initial data are not sufficient to identify a unique solution, which can be obtained only if one supplements with a suitable boundary condition. From the perspective of fundamental solutions, this statement amounts to saying that one cannot expect unique, advanced and retarded Green's operators, rather their existence should be tied to the specific boundary conditions assigned on $\partial N$. The goal of this paper is to start a full-fledged analysis of this problem on a large class of manifolds with non-empty boundary. Observe that, although we restrict the attention to the scalar wave equation, a similar analysis could be done for vector valued partial differential equations. In the special case of the Dirac field a recent work, connected to this problem, has appeared in \cite{GrosseMurro}.

We focus on globally hyperbolic, oriented and time-oriented Lorentzian manifolds $(N,h)$ with timelike boundary \cite{Ake-Flores-Sanchez-18}, such that there exists, in addition, an isometry between $(N,h)$ and a standard static spacetime $M\times_{\beta}\mathbb{R}$. This is a warped product between $\mathbb{R}$ and a Riemannian manifold $(M,g)$ with a non-empty boundary $(\partial M,\iota^*_Mg)$ where $\iota_M:\partial M\hookrightarrow M$. Thereon, we consider the D'Alembert wave operator $\widetilde\square$ as well as the auxiliary problem $\square=-\partial^2_t+A$, $t$ being the coordinate along $\mathbb{R}$, while $A$ is a second-order, elliptic, partial differential operator built out of $g$ and of $\beta$.
In this framework the advanced and retarded fundamental fundamental solutions are defined as $\mathcal{G}^\pm\in\mathcal{D}^\prime(\mathring{N}\times\mathring{N})$, $\mathring{N}\doteq N\setminus\partial N$, whose associated Green operators $\mathsf{G}^\pm:\mathcal{D}(\mathring{N})\to\mathcal{D}^\prime(\mathring{N})$ satisfy $\square\circ\mathsf{G}^\pm=\mathsf{G}^\pm\circ\square=\textrm{id}|_{\mathcal{D}(\mathring{N})}$, while also enjoying the support property $\textrm{supp}(\mathsf{G}^\pm(f))\subseteq J^\pm(\textrm{supp}(f))$, for all $f\in\mathcal{D}(\mathring{N})$. On account of the underlying metric being static, as a first step one can reduce the problem to an auxiliary one obtained via the ansatz $\mathcal{G}^+=\theta(t-t^\prime)\mathcal{G}$ and $\mathcal{G}^-=-\theta(t^\prime-t)\mathcal{G}$, where $\theta$ is the Heaviside distribution while $\mathcal{G}\in\mathcal{D}^\prime(\mathring{N}\times\mathring{N})$ is a solution of 
\begin{equation*}
\left\{\begin{array}{l}
(\square\otimes\mathbb{I})\mathcal{G}=(\mathbb{I}\otimes\square)\mathcal{G}=0\\
\mathcal{G}|_{t=t^\prime}=0,\quad\partial_t\mathcal{G}|_{t=t^\prime}=-\partial_{t^\prime}\mathcal{G}|_{t=t^\prime}=\delta_{\textrm{diag}(\mathring{M})}
\end{array}\right. ,
\end{equation*}
where $|_{t=t^\prime}$ indicates the pull-back of a distribution on $\mathring{N}\times\mathring{N}$ to $(\{t\}\times\mathring{M})\times (\{t\}\times\mathring{M})$ while $\delta_{\textrm{diag}(\mathring{M})}$ is the Dirac delta distribution on the diagonal of $\mathring{M}\times\mathring{M}$. As it stands, the above initial value problem is incomplete unless one also assigns a boundary condition on $\partial N$ which allows for the identification of a unique solution. In order to understand which class of such conditions is admissible we focus our attention on $A$, which, being the metric static, is a partial differential operator acting on scalar functions over $M$.
Our strategy proceeds in two steps. In the first we start by reading $A$ as a symmetric operator on $\mathsf{H}\equiv L^2(M;\textrm{d}\mu_g)$, $\textrm{d}\mu_g$ being the metric-induced volume form and we show that, for every self-adjoint extension of $A$, by means of spectral calculus, it is possible to construct $\mathcal{G}$, solution of the above system of partial differential equations. 

In the second step instead we start by investigating which are the possible self-adjoint extensions of $A$ and how to characterize them explicitly in terms of boundary conditions. To this avail we require an additional assumption, namely $(M,g)$ ought to be of bounded geometry, a widely used structure which has been recently studied in the context of manifolds with boundary in \cite{AmmanGrosseNistor,GrosseSchneider,Schick}. On the one hand, we observe that our assumptions include several interesting scenarios such as anti-de Sitter spacetime. On the other hand, Riemannian manifolds with boundary and of bounded geometry turn out to be the natural framework to construct Sobolev spaces and to prove thereon a generalization of the standard Lions-Magenes trace theorems on $\mathbb{R}^n$. This last feature is of special interest to us since it allows us to classify and to characterize the self-adjoint extensions using the framework of boundary triples, \cite{Grubb68} and \cite[Ch. 6]{HSF12}.

In short and adapting it to the case at hand, the latter includes and even supersedes Von Neumann theory of deficiency indexes -- see \cite{Posilicano-18} for an equivalent approach based on resolvents.
It associates to the symmetric operator $A$ an auxiliary Hilbert space, in our case chosen as $\mathsf{h}\equiv L^2(\partial M, \textrm{d}\mu_{\iota^*_Mg})$, together with two linear, surjective maps $\gamma_i\colon D(A^*)\to\mathsf{h}$, $i=0,1$ where $A^*$ is the adjoint of $A$, so that a Green's formula holds true:
\begin{gather*}
(A^*f|f^\prime)_{\mathsf{H}}-(f|A^*f^\prime)_{\mathsf{H}}=
(\gamma_1f|\gamma_0f^\prime)_{\mathsf{h}}-(\gamma_0f|\gamma_1f^\prime)_{\mathsf{h}}\qquad
\forall f,f^\prime\in D(A^*).
\end{gather*}
In the case we are interested in, it turns out that $\gamma_0,\gamma_1$ can be built working at the level of Sobolev spaces over $M$ as combination between the above mentioned trace theorems and the normal derivative to $\partial M$.
As a by-product, one can apply a standard result from the theory of boundary triples according to which there is a one-to-one correspondence between the self-adjoint operators $\Theta$ on $\mathsf{h}$ and the self-adjoint extensions of $A$, denoted by $A_\Theta$.

By using these results and the spectral calculus for $A_\Theta$ we are able to construct explicitly, for each choice of $\Theta$, fundamental solutions $\mathcal{G}^\pm_\Theta$ associated to $\square$, proving in addition with an energy estimate that they enjoy the sought support properties.
In addition, we show that, for every $f\in\mathcal{D}(N)$, it holds that $\gamma_1(\mathsf{G}_\Theta^\pm(f))=\Theta\gamma_0(\mathsf{G}_\Theta^\pm(f))$, where $\mathsf{G}^\pm_\Theta$ are the advanced and retarded Green's operators associated with $\mathcal{G}^\pm_\Theta$.
This justifies our statement according to which fixing $\Theta$, and thus a pair of fundamental solutions $\mathcal{G}^\pm_\Theta$, is subordinated to a choice of boundary condition. Furthermore, under mild assumptions on the warping factor $\beta$, so to ensure that the operator $A$ is uniformly elliptic, we are able to prove two additional results which strongly characterize the advanced and retarded Green's operators. On the one hand, we show that $\mathsf{G}^\pm_\Theta[\mathcal{D}(\mathring{N})]\subset C^\infty(N)$, on the other hand we are able to enlarge the domain of $\mathsf{G}_\Theta\doteq\mathsf{G}^-_\Theta-\mathsf{G}^+_\Theta$ so to construct an exact sequence of linear maps, which characterize completely, for every admissible $\Theta$, the space of associated smooth solutions of the equation of motion ruled by $\square$. 

Finally we observe that, although the class of boundary conditions that we define in terms of the self-adjoint extensions of $A$ is rather large including for example those of Robin type, it does not encompass some cases which are often discussed in the literature. Most notably it might be desirable to allow an explicit time-dependence of the boundary condition, as it happens for example in those of Wentzell type, which have remarkable applications in several models, see \cite{Feller57,Ueno,Zahn18} and \cite{Gal} in particular. The last part of this work is devoted to an analysis of such scenario. We show that it is possible to extend the range of applicability of the framework of boundary triples proving in particular the existence of fundamental solutions for a larger class of boundary conditions, including those of Wentzell type.

\vskip .2cm

The structure of the paper is the following: In Section \ref{Sec:Geometric_data} we introduce the basic geometric data of this paper. In particular, we emphasize the notion of a static Lorentzian spacetime with timelike boundary and we outline the notion of manifold with boundary and of bounded geometry, reviewing its main properties. In Section \ref{Sec:Sobolev_spaces} we introduce Sobolev spaces on Riemannian manifolds of bounded geometry and we put a particular emphasis on the trace theorem as proven in \cite{GrosseSchneider}. In Section \ref{Sec:boundary_triples} we discuss the framework of boundary triples, first from an abstract point of view and then we specialize it to the case of a second-order, elliptic partial differential operators, making a direct connection to the theory of Sobolev spaces outlined in the previous section. In particular, we use boundary triples to characterize the self-adjoint extensions of symmetric operators with equal deficiency indexes. In Section \ref{Sec:Fundamental_Solutions} we obtain the main results of this work. Here we start from the D'Alembert wave operator on a static Lorentzian spacetime and we rewrite it in terms of an equivalent operator of the form $\square=-\partial^2_t+A$, where $A$ is a second-order, elliptic partial differential operator, symmetric on $L^2(M;\textrm{d}\mu_g)$. As a first step, by using spectral analysis, we prove the existence of advanced and retarded Green's operators for any choice of self-adjoint extension of $A$. Via a suitable boundary triple, these are in turn put into a one-to-one relation with the choice of a self-adjoint operator on $L^2(\partial M;\textrm{d}\mu_{i^*_Mg})$. Subsequently, we prove several structural properties of the fundamental solutions, ranging from the support  to the characterization of the relation between the choice of boundary condition for $\square$ and the self-adjoint extension of $A$. Furthermore, under additional mild conditions on the metric, we show that all smooth solutions to the wave equation can be written in terms of the advanced and retarded Green's operators which thus turn out to encompass as much information on the underlying operator ruling the dynamics as their natural counterpart on globally hyperbolic spacetimes. At last, in Section \ref{Sec:Dynamical_boundary_condition}, we extend the previous framework to account also for a larger class of time-dependent boundary condition, including in particular those of Wentzell type.

\subsection{Geometric data}\label{Sec:Geometric_data}

The goal of this section is to introduce both the geometric data and the key functional spaces at the heart of this work, fixing in particular the notations and conventions. With respect to the structure of Lorentzian manifolds with empty boundary we refer mainly to \cite{BEE}, while, for the case with a non-empty, timelike boundary, recent analyses are available in \cite{CGS,Solis}. It is noteworthy that our framework can be read as a special instance of that considered and studied in \cite{Ake-Flores-Sanchez-18}.

 Following the standard definition, see for example \cite[Ch. 1]{Lee}, in this paper both the symbols $M$ and $N$ refer to a smooth, second-countable, connected, manifold of dimension $m\geq 1$ ({\em resp.} $m+1$), with smooth boundary $\partial M$ ({\em resp.} $\partial N$). A point $p\in M$ ({\em resp.} $N$) such that there exists an open neighbourhood $U$ containing $p$, diffeomorphic to an open subset of $\mathbb{R}^m$ ({\em resp.} $\mathbb{R}^{m+1}$), is called an {\em interior point} and the collection of these points is indicated with $Int(M)\equiv\mathring{M}$ ({\em resp.} $Int(N)\equiv\mathring{N}$). As a consequence $\partial M$ ({\em resp.} $\partial N$), if non empty, can be read as a manifold on its own and $\partial M=M\setminus\mathring{M}$ ({\em resp.} $\partial N=N\setminus\mathring{N}$). 
 
\begin{Definition}\label{Def:Lorentzian_manifold}
	We say that $N$ is a {\bf Lorentzian manifold with timelike boundary} if it is oriented, time oriented and endowed with a smooth Lorentzian metric $h$ such that also $\iota_N^*h$ is a Lorentzian metric, $\iota_N$ being the embedding of $\partial N$ in $N$.
\end{Definition}

In the class of Lorentzian manifolds with timelike boundary $(N,h)$, we will be interested in those which are {\em standard static}, that is for which there exists a nowhere vanishing, irrotational, timelike Killing field $\chi\in\Gamma(TN)$, {\it cf.} \cite[Lemma 3.78]{BEE} and \cite{San06}, and $(N,h)$ is isometric to the warped product $M\times_\beta \mathbb{R}$, with line element 
\begin{equation}\label{eq:line_element}
h=-\beta \textrm{d}t^2+g\,,
\end{equation}
where $t\colon M\times_\beta \mathbb{R}\to\mathbb{R}$ is the projection along the second component, playing thus the role of time variable, $\beta\in C^\infty(M;(0,\infty))$ and $g$ identifies a time-independent Riemannian metric on each submanifold $\{t\}\times M$. As a consequence 

\begin{corollary}\label{cor:standard_static}
	Let $(N,h)$ be a standard static, Lorentzian manifold with timelike boundary. Then also $\partial N$ is a standard static Lorentzian manifold. 
\end{corollary}

\begin{proof}
	Per hypothesis $(N,h)$ is isometric to $M\times_\beta \mathbb{R}$ with line element \eqref{eq:line_element} where $\partial_t$ plays the role of the complete, irrotational, timelike, nowhere vanishing Killing field. Hence $\partial N$ is isometric to $\partial M\times_\beta\mathbb{R}$ and, calling $\iota_M:\partial M\to M$ the embedding map, \eqref{eq:line_element} reduces to $-\widetilde{\beta}\textrm{d}t^2+\iota^*_M(g)$, $\widetilde{\beta}\doteq\beta|_{\partial M}$. As a consequence $\partial N$ has the sought property.
\end{proof}

Observe that, with these data, $M$ comes equipped with an induced orientation and with a smooth Riemannian metric $g$, so that $(\partial M, \widetilde{g})$ is also an oriented Riemannian manifold if endowed with the induced orientation and metric, that is $\widetilde{g}\doteq \iota^*_M(g)$, $\iota_M$ being the embedding of $\partial M$ in $M$. In the class of Riemannian manifolds with non-empty boundary, we are interested in a particular subclass, distinguished by its geometric properties in a neighbourhood of $\partial M$. The following definitions were first given in \cite{AmmanGrosseNistor,GrosseNistor}, barring the next one, which is standard, {\it cf.} \cite{Eich91}:

\begin{Definition}\label{Def:bounded_geometry}
	A Riemannian manifold $(M,g)$ with $\partial M=\emptyset$ is called of {\bf bounded geometry} if the injectivity radius $r_{\textrm{inj}}(M)>0$ and if $TM$ is of {\em totally bounded curvature}, that is $\|\nabla^k R\|_{L^\infty(M)}<\infty$ for all $k\in\mathbb{N}\cup\{0\}$, $R$ being the scalar curvature and $\nabla$ the Levi-Civita connection associated with $g$.
\end{Definition}

\noindent To avoid the problem that $r_{\textrm{inj}}(M)$ vanishes whenever $\partial M\neq\emptyset$, one must first consider the following generalization:

\begin{Definition}\label{Def:bounded_geometry_submanifold}
	Let $(M,g)$ be a Riemannian manifold of bounded geometry and let $(Y,\iota_Y)$ be a codimension $k$ closed, embedded, smooth submanifold with an inward pointing, unit normal vector field $\nu_Y$. We say that $(Y,\iota^*_Y g)$ is a {\em bounded geometry submanifold} if the following holds:
	\begin{enumerate}
		\item the second fundamental form $K_Y$ of $Y$ in $M$ together with all its covariant derivatives along $Y$ is bounded,
		\item there exists $\epsilon_Y>0$ such that the map $\varphi_{\nu_Y}:Y\times(-\epsilon_Y,\epsilon_Y)\to M$ defined as $(x,z)\mapsto\varphi_{\nu_Y}(x,z)\doteq\exp_x(z\nu_{Y,x})$ is injective, where $\exp_x$ is the exponential map of $M$ at $x$.
	\end{enumerate}
\end{Definition} 

\noindent We observe that, as proven in \cite{AmmanGrosseNistor}, Definition \ref{Def:bounded_geometry_submanifold} entails that $(Y,\iota^*_Y g)$ is automatically a Riemannian manifold of bounded geometry. We introduce the class of Riemannian manifolds we will be working with in this paper:

\begin{Definition}\label{Def:bounded_geometry_boundary}
	Let $(M,g)$ be a Riemmannian manifold with non-empty boundary $\partial M$. We say that $(M,g)$ has {\bf bounded geometry} if there exists a Riemannian manifold of bounded geometry $(\widehat{M},\widehat{g})$ of the same dimension of $M$ such that 
	\begin{enumerate}
		\item $M\subset\widehat{M}$ and $g=\widehat{g}|_M$,
		\item $(\partial M,\iota^*\widehat{g})$ is a bounded geometry submanifold of $\widehat{M}$, where $\iota:\partial M\to\widehat{M}$ is the embedding map\footnote{Recall that, if we consider the restriction map $\textrm{res}:\widehat{M}\to M$, then $\textrm{res}\circ\iota=\iota$. Hence $\iota^*\widehat{g}=\iota^*g$.}.
	\end{enumerate}
\end{Definition}

\begin{remark}
	Definition \ref{Def:bounded_geometry_boundary} is equivalent to the original one of manifolds with boundary and of bounded geometry given by Schick in \cite{Schick}.
	We observe that, while the definition given in this last cited paper does not require any extrinsic data such as in particular $\widehat{M}$, all results obtained are independent from the choice of the latter. 
	
	 Since, from now on, we will be working only with Riemannian manifolds with non-empty boundary and of bounded geometry, we will drop the subscript $Y$ as in Definition \ref{Def:bounded_geometry_submanifold}, since we will be always referring to $\partial M$ as the embedded submanifold. In addition we call {\em geodesic collar} (of $\partial M$) the set $\partial M\times [0,\epsilon)$ such that the map $\varphi_\nu$ is a diffeomorphism onto its image and we define
	\begin{equation}\label{eq:geodesic_collar}
	\mathcal{GC}_\epsilon(M)\doteq\varphi_\nu[\partial M\times [0,\epsilon)].
	\end{equation}
\end{remark}

\begin{remark}
	We observe that all Riemannian manifolds with compact boundary meet the requirements of Definition \ref{Def:bounded_geometry_boundary}. At the same time one can also consider non-compact boundaries such as for example $\mathbb{H}^n$, the collection of all points $(x_1,...,x_n)\in\overline{\mathbb{R}_+}\times\mathbb{R}^{n-1}$ endowed with the Euclidean metric of $\mathbb{R}^n$.
\end{remark}

As a next step, if $(M,g)$ is a Riemannian manifold of bounded geometry, we can introduce a distinguished set of coordinates which are at the heart of the definition of Sobolev spaces and of the associated trace theorem proven in \cite{GrosseSchneider}. Here, we will recall only the basic structures and facts, leaving all details and proofs to \cite[Sec. 4.2]{AmmanGrosseNistor} and to \cite[Sec. 4.1]{GrosseSchneider}. Note that the following construction was used in \cite{Schick} though with the name of normal collar coordinates.

Let $(M,g)$ be a Riemannian manifold with boundary and of bounded geometry as per Definition \ref{Def:bounded_geometry_boundary}. For any $p\in\partial M$, we can choose any orthonormal basis of $T_p\partial M$ to identify it with $\mathbb{R}^{m-1}$, $m=\dim M$. From now on this identification will be left implicit. Since the injectivity radius of $\partial M$ is finite, for all $0<r<r_{\textrm{inj}}(\partial M)$, the exponential map $\exp_p^{\partial M}:\mathcal{B}^{m-1}(0)\to\mathcal{B}_r(p)$ is a diffeomorphism. Here $\mathcal{B}^{m-1}(0)$ stands for the ball of radius $r$ in $\mathbb{R}^{m-1}$ centered at $0$. By considering in addition the map $\varphi_\nu$ for $\partial M$ (see Definition \ref{Def:bounded_geometry_submanifold} and \eqref{eq:geodesic_collar}), whenever $0<r<\min\left\{\frac{r_{\textrm{inj}}(\partial M)}{2},\frac{r_{\textrm{inj}}(M)}{4},\frac{\epsilon}{2}\right\}$, we identify the following:
\begin{equation}\left\{ 
\begin{array}{lll}
\kappa_p:\mathcal{B}^{m-1}(0)\times [0,r) \to M, &\kappa_p(x,z)\doteq\varphi_\nu(\exp_p^{\partial M}(x),z), & p\in\partial M\\
\kappa_p:\mathcal{B}_r^m(0)\to M, & \kappa_p(v)\doteq\exp^M_p(v), & p\in\mathring{M}
\end{array}\right.,
\end{equation}
where $\mathcal{B}_r^m(0)$ indicates the ball of radius $r$ centered at the origin of $T_pM$, here implicitly identified with $\mathbb{R}^m$. Let $U_p(r)$ stand for the image in $M$ of the map $\kappa_p$, then we can define the following:

\begin{Definition}\label{Def:Fermi_coordinates}
	Let $(M,g)$ be a Riemannian manifold with boundary and of bounded geometry and let $0<r<\min\left\{\frac{r_{\textrm{inj}}(\partial M)}{2},\frac{r_{\textrm{inj}}(M)}{4},\frac{\epsilon}{2}\right\}$. We call {\em Fermi (or normal collar) chart} the map $\kappa_p:\mathcal{B}^{m-1}(0)\times[0,r)\to M$ for $p\in\partial M$. The ensuing coordinates $(x^i,z):U_p(r)\to\mathbb{R}^{m-1}\times [0,\infty)$, $i=1,...,m-1$ are called {\em Fermi (or normal collar) coordinates}.
\end{Definition}

If the point $p$ does not lie on the boundary of $M$, one can adapt straightforwardly this last definition to obtain the standard normal geodesic coordinates. Since we will not make use of them, we omit giving an explicit expression. 

To conclude the section, we study the interplay between Riemannian manifolds with boundary and of bounded geometry and standard static, Lorentzian manifolds with a timelike boundary. 

\begin{proposition}\label{prop:static_spacetime_with_boundary}
	Let $(M,g)$ be a Riemannian manifold with boundary and of bounded geometry and let $(\widehat{M},\widehat{g})$ be a Riemannian manifold of bounded geometry such that $M\subset\widehat{M}$ and $g=\widehat{g}|_M$. Then
	\begin{enumerate}
		\item Every $\beta\in C^\infty(M;(0,\infty))$ identifies an isometry class $[(N,h)]$ of standard static, Lorentzian manifolds with timelike boundary,
		\item If in addition there exists $\widehat{\beta}\in C^\infty(\widehat{M};(0,\infty))$ such that $\widehat{\beta}|_M=\beta$ and if $\frac{\widehat{g}}{\widehat{\beta}}$ identifies a complete Riemannian metric on $\widehat{M}$, then each representative $(N,h)$ of $[(N,h)]$ is a submanifold of a standard static, globally hyperbolic spacetime $(\widehat{N},\widehat{h})$. 
	\end{enumerate}
A manifold $(N,h)$ satisfying condition $1.$ of Proposition \ref{prop:static_spacetime_with_boundary} will be called a {\em static Lorentzian spacetime, with timelike boundary}.
\end{proposition}

\begin{proof}
	Consider $M$ as per hypothesis and construct the warped product $M\times_\beta\mathbb{R}$ endowed with the line element $\textrm{d}s^2=-\beta \textrm{d}t^2+g$ where $t\colon M\times_\beta\mathbb{R}\to\mathbb{R}$ is the projection along the second component. Every manifold $(N,h)$ which is isometric to $M\times_\beta\mathbb{R}$ with the given metric is standard static, hence proving the first point. To prove the second statement, it suffices to observe that $M\times_\beta\mathbb{R}$ can be seen as being isometrically embedded in the standard static spacetime $\widehat{M}\times_{\widehat{\beta}}\mathbb{R}$ with line element $\textrm{d}s^2=-\widehat{\beta}\textrm{d}t^2 + \widehat{g}$. This manifold is globally hyperbolic on account of \cite[Theorem 3.66]{BEE}.
\end{proof}

\begin{remark}
	Observe that condition 2 in the last proposition is a constraint only on the admissible functions $\widehat{\beta}$. As a matter of fact, every Riemannian manifold of bounded geometry is metric complete \cite{Eich91} and thus $\frac{\widehat{g}}{\widehat{\beta}}$ is also complete if and only if $\widehat{\beta}$ behaves at most quadratically at infinity, {\it cf.} \cite[Rem. 2.2]{San06}.
\end{remark}

\subsection{Sobolev spaces on manifolds of bounded geometry}\label{Sec:Sobolev_spaces}

We shall introduce the functional spaces that we will be using in the next sections, as well as their main properties. We will be using most of the results from \cite{GrosseSchneider}.
In the following we consider $(\widehat{M},\widehat{g})$, a Riemannian manifold of bounded geometry such that $\partial\widehat{M}=\emptyset$.
The case with non empty boundary has been discussed mainly in \cite{AmmanGrosseNistor}.
With $\mathcal{D}(\widehat{M})$ we will indicate the space of smooth, compactly supported functions on $\widehat{M}$ endowed with the standard locally convex topology, while with $L^p(\widehat{M})\doteq L^p(\widehat{M},\textrm{d}\mu_{\widehat{g}})$, $p\in\mathbb{N}$ we consider the completion of $\mathcal{D}(\widehat{M})$ with respect to the $L^p$-norm constructed out of the metric induced volume form $\textrm{d}\mu_{\widehat{g}}$.
With $\mathcal{E}(\widehat{M})$ we will indicate the space of smooth functions on $\widehat{M}$ endowed with the standard locally convex topology.
With $\mathcal{D}^\prime(\widehat{M})$ we refer to the space of distributions, whose test functions are the elements of $\mathcal{D}(\widehat{M})$.

\begin{remark}\label{Rem:quotient_Spaces}
The same definitions apply, mutatis mutandis, to the case of $(M,g)$ being a Riemannian manifold with boundary and of bounded geometry, though in this case $\mathcal{D}(M)$ is replaced in the preceding and in the forthcoming discussion by $C_{me}(M)$, that is the equivalence classes of complex valued measurable functions over $M$. Observe that, in view of Definition \ref{Def:bounded_geometry_boundary}, we can replace $\mathcal{D}(M)$ also with  $\mathcal{D}(\widehat{M})/\{f\in\mathcal{D}(\widehat{M})|\; f|_M=0\}$, which is isomorphic to the former.
\end{remark}

In order to introduce Sobolev spaces we need suitable local charts. On the one hand, since every Riemannian manifold $(\widehat{M},\widehat{g})$ of bounded geometry is also complete, we can define the standard geodesic normal coordinates, whose associated atlas is indicated with $\big(U^{\textrm{geo}}_\beta,\kappa^{\textrm{geo}}_\beta\big)_{\beta\in J}$, $J$ being a suitable index set. If we let $\{h^{\textrm{geo}}_\beta\}_{\beta\in J}$ be a partition of unity subordinated to this cover we have identified the triple $\mathcal{T}^{\textrm{geo}}\doteq\big(U^{\textrm{geo}}_\beta,\kappa^{\textrm{geo}}_\beta, h_\beta\big)_{\beta\in J}$, which we will refer to as {\em geodesic trivialization} of $\widehat{M}$.

On the other hand we say that a cover $\{U_\alpha\}_{\alpha\in I}$ of $\widehat{M}$, $I$ being an index set, is {\em uniformly locally finite} if there exists $n\in\mathbb{N}$ such that each element of the cover is intersected by at most $n$ other sets of the cover. In addition, we consider 
\begin{enumerate}
	\item on each $U_\alpha$, $\alpha\in I$, local coordinates, that is a diffeomorphism $\kappa_\alpha: V_\alpha\to U_\alpha$, where $V_\alpha$ is an open subset of $\mathbb{R}^m$, $m=\dim M$,
	\item a partition of unity $\{h_\alpha\}_{\alpha\in I}$ subordinated to the cover $\{U_\alpha\}_{\alpha\in I}$.
\end{enumerate}
The triple $\mathcal{T}\doteq\left(U_\alpha,\kappa_\alpha,h_\alpha\right)_{\alpha\in I}$ is called a {\em uniformly locally finite trivialization of $\widehat{M}$}. In the collection of these trivializations, we select a distinguished subclass by the relation with geodesic coordinates:
\begin{Definition}\label{Def:admissible_trivialization}
	Let $(\widehat{M},\widehat{g})$ be a manifold of bounded geometry of dimension $m$. We call $\mathcal{T}$ a uniformly locally finite trivialization of $\widehat{M}$ {\em admissible} if the following two conditions are met:
	\begin{enumerate}
		\item The atlas $(U_\alpha,k_\alpha)_{\alpha\in I}$ built out of $\mathcal{T}$ is compatible with a geodesic atlas $(U_\beta^{\textrm{geo}},k_\beta^{\textrm{geo}})_{\beta\in J}$ of $\widehat{M}$, that is, for all $k\in\mathbb{N}\cup\{0\}$, there exists $C_k>0$ such that, for all $\alpha\in I$, for all $\beta\in J$ and for all multi-indices $\mathfrak{a}\in(\mathbb{N}\cup\{0\})^m$ with $|\mathfrak{a}|\leq k$,
		$$\left|D^{\mathfrak{a}}(\kappa_\alpha^{-1}\circ\kappa^{\textrm{geo}}_\beta)\right|\leq C_k\quad\textrm{and}\quad\left|D^{\mathfrak{a}}((\kappa_\beta^{\textrm{geo}})^{-1}\circ\kappa_\alpha)\right|\leq C_k,$$
		\item for all $k\in\mathbb{N}$, there exists $c_k>0$ such that, for all $\alpha\in I$ and for all multi-indices $\mathfrak{a}\in(\mathbb{N}\cup\{0\})^m$ with $|\mathfrak{a}|\leq k$, it holds
		$$\left|D^{\mathfrak{a}}(h_\alpha\circ\kappa_\alpha)\right|\leq c_k.$$
	\end{enumerate} 
\end{Definition}

\noindent From now on we shall only consider admissible, uniformly locally finite trivializations.

\begin{Definition}\label{Def:Sobolev_space}
	Let $(\widehat{M},\widehat{g})$ be a Riemannian manifold of bounded geometry of dimension $m$ and let $\mathcal{T}$ be an admissible, uniformly locally finite trivialization, and let $\mathcal{T}^{\textrm{geo}}$ be an associated geodesic trivialization. Then, for every $s\in\mathbb{R}$ and for every integer $1<p<\infty$, we call $H^{s,\mathcal{T}}_p(\widehat{M})$ the collection of all distributions $u\in\mathcal{D}^\prime(\widehat{M})$ such that
	$$\|u\|_{H^{s,\mathcal{T}}_p}\doteq\Bigg[\sum_{\alpha\in I}\|(h_\alpha u)\circ\kappa_\alpha\|^p_{H^s_p(\mathbb{R}^m)}\Bigg]^{\frac{1}{p}}<\infty,$$
	where $\|\cdot\|_{H^s_p(\mathbb{R}^m)}$ indicates the standard Sobolev norm on $_{H^s_p(\mathbb{R}^m)}$. Equivalently we define $H^{s,\mathcal{T}^{\textrm{geo}}}_p(\widehat{M})$ by replacing $\mathcal{T}$ with $\mathcal{T}^{\textrm{geo}}$.
\end{Definition}

The following proposition summarizes the results of \cite[Th. 3.9]{GrosseSchneider} and of Section 7.4.5 in \cite{Triebel_vol2}, see also \cite{Hebey}:

\begin{proposition}\label{Prop:Sobolev_space_equivalence}
	Let $(\widehat{M},\widehat{g})$ be a Riemannian manifold of bounded geometry and let $\mathcal{T},\mathcal{T}^{\textrm{geo}}$ be respectively a uniformly locally finite and a geodesic trivialization of $\widehat{M}$. Let $k\in\mathbb{N}\cup\{0\}$ and let
	$W^k_p(\widehat{M})$ be the completion of the subspace $\mathcal{E}^k_p(\widehat{M})\doteq\{f\in\mathcal{E}(\widehat{M})|\quad f,\nabla f,\ldots,\nabla^kf\in L^p(\widehat{M}))\}$ with respect to the norm 
	$$\|f\|_{W^k_p(\widehat{M})}\doteq\bigg(\sum_{i=0}^k\|\nabla^i f\|^p_{L^p(\widehat{M})}\bigg)^{\frac 1p}\,,$$
	$\nabla$ being the covariant derivative built out of the metric $\widehat{g}$. Then it holds that, for all $s\in\mathbb{R}$ and for all integer $p$ such that $1<p<\infty$, neither $H^{s,\mathcal{T}}_p(\widehat{M})$ nor $H^{s,\mathcal{T}^{\textrm{geo}}}_p(\widehat{M})$ depend on the choice of the trivialization. Hence $H^s_p(\widehat{M})\equiv H^{s,\mathcal{T}}_p(\widehat{M})=H^{s,\mathcal{T}^{\textrm{geo}}}_p(\widehat{M})$. In addition, if $s\in\mathbb{N}\cup\{0\}$, it holds that $W^s_p(\widehat{M})=H^s_p(\widehat{M})$.
\end{proposition}

\begin{remark}\label{Remark: Sobolev spaces on manifold with boundary}
In the preceding discussion we have considered only manifolds without boundary.
Nonetheless it is possible to extend or to adapt all definitions also to any Riemannian manifold $(M,g)$ with boundary and of bounded geometry using Definition \ref{Def:Sobolev_space}.
A detailed discussion has been given especially in \cite[Sec. 5.1]{AmmanGrosseNistor}.
In particular we observe that Fermi coordinates as per Definition \ref{Def:Fermi_coordinates} can be completed to define an admissible trivialization out of which it is possible to define $H^{s,\mathcal{T}}_p(M)$ for all $s\in\mathbb{R}$ and for all integer values of $p$ such that $1<p<\infty$.
Most notably, it holds, also in this case, that $H^s_p(M)=W^s_p(M)$ for all $s\in\mathbb{N}\cup\{0\}$.

Whenever a boundary is present, one can introduce the subspace $H^s_{0,p}(M)\subset H^s_p(M)$ defined as the completion of $\mathcal{D}(M)$ with respect to the $H^s_p(M)$-norm.
Observe that, whenever $M$ is metric complete ({\it e.g.} if $M$ is a Riemannian manifold of bounded geometry), $H^s_{0,p}(M)=H^s_p(M)$, while the inclusion is strict in general. 
\end{remark}

\noindent To conclude this section we state the trace theorem, as proven in \cite[Th. 4.10 \& Cor. 4.12]{GrosseSchneider}, though specialized to the case of our interest:

\begin{theorem}\label{Th:trace}
Let $(M,g)$ be a Riemannian manifold with boundary and of bounded geometry. Let $s\in\mathbb{R}$ and let $H^s(M)\doteq H^s_2(M)$, $H^s_0(M)\doteq H^2_{0,2}(M)$ be the Sobolev spaces as per Definition \ref{Def:Sobolev_space} and per Proposition \ref{Prop:Sobolev_space_equivalence}, see also Remark \ref{Remark: Sobolev spaces on manifold with boundary}. Then, if $s>\frac{1}{2}$, the restriction map from $\mathcal{D}(M)$ to $\mathcal{D}(\partial M)$ extends to a continuous surjective map 
$$\Gamma:H^s(M)\to H^{s-\frac{1}{2}}(\partial M)\,.$$
If $s\geq 1$, $\ker\Gamma=H^s(M)\cap H^1_0(M)$.
\end{theorem}

\noindent Observe that, in view of this last theorem, one can read the elements of $H^s_{0,p}(M)$, $s\in\mathbb{N}$ as those of $H^s_p(M)$ whose representatives are functions whose derivatives with order less than $s$ have vanishing trace on $\partial M$.

\section{Boundary triples and their application to second-order elliptic differential operator}\label{Sec:boundary_triples}

In this section, we consider second-order, elliptic differential operators on a Riemannian manifold with boundary and of bounded geometry, and characterize their self-adjoint extensions. A convenient framework for addressing this question is that of {\em boundary triples}, a thoroughly analysed topic of which we recall the main aspects following \cite{BL12} and references therein. 

\subsection{Boundary triples}

In this section $\mathsf{H}$ indicates a separable Hilbert space over $\mathbb{C}$ while $S\colon D(S)\subset\mathsf{H}\to\mathsf{H}$ is a closed, symmetric linear operator. 
\begin{Definition}\label{Definition: boundary triples}
A {\em boundary triple} for the adjoint operator $S^*$ is a triple $(\mathsf{h},\gamma_0,\gamma_1)$ consisting of a separable Hilbert space $\mathsf{h}$ over $\mathbb{C}$ and two linear maps $\gamma_0,\gamma_1\colon D(S^*)\to\mathsf{h}$ such that
\begin{gather}\label{Equation: boundary equation for boundary triples}
(S^*f|f^\prime)_{\mathsf{H}}-(f|S^*f^\prime)_{\mathsf{H}}=
(\gamma_1f|\gamma_0f^\prime)_{\mathsf{h}}-(\gamma_0f|\gamma_1f^\prime)_{\mathsf{h}}\qquad
\forall f,f^\prime\in D(S^*)\,.
\end{gather}
and the map $\gamma:D(S^*)\to\mathsf{h}\times\mathsf{h}$, defined by $\gamma(f)=(\gamma_0(f),\gamma_1(f))$ for all $f\in D(S^*)$, is surjective.
\end{Definition}

\begin{Example}\label{Ex:ordinary_boundary_triple}
A canonical example of boundary triple, sometimes termed ordinary boundary triple, can be constructed if $S$ has equal and finite deficiency indices $d_\pm(S)=\dim(\mathcal{N}_\pm(S^*))<\infty$, $\mathcal{N}_\pm(S^*)\doteq\ker(S^*\pm i\mathbb{I})$, hence admitting self-adjoint extensions. In this case, letting $V\colon\mathcal{N}_-(S^*)\to\mathcal{N}_+(S^*)$ be any, but fixed unitary operator, one can set $\mathsf{h}=\mathcal{N}_+(S^*)$. Since $D(S^*)=D(S)\oplus_S\mathcal{N}_+(S^*)\oplus_S\mathcal{N}_-(S^*)$ \footnote{The orthogonal decomposition $\oplus_S$ refers to the scalar product on $D(S^*)$ defined by $(f\,|\,f')_S:=(f\,|\,f')+(S^*f\,|\,S^*f')\,.$}, one can introduce the projection maps $\pi_\pm:D(S^*)\to\mathcal{N}_\pm(S^*)$, defining 
$\gamma_0\doteq \pi_+ - V\circ\pi_-$ and $\gamma_1\doteq i(\pi_+ + V\circ\pi_-)$. Observe that, since $V$ is a bijection, the map $\gamma$ as in Definition \ref{Definition: boundary triples} is automatically surjective.
Finally, equation \eqref{Equation: boundary equation for boundary triples} follows by the identity
$$(S^*f|f^\prime)-(f|S^*f^\prime)=2i\big[(\pi_-f|\pi_-f^\prime)_{\mathsf{H}}-(\pi_+f|\pi_+f^\prime)_{\mathsf{H}}\big]\,,$$
which holds true for all $f,f^\prime\in D(S^*)$.
\end{Example}

\noindent Boundary triples allow us to characterize the self-adjoint extensions of $S$ in terms of ``boundary conditions'' on $\mathsf{h}$. As a starting point one can observe, that, under the hypotheses of Definition \ref{Definition: boundary triples}, one can single out two distinguished self-adjoint extensions of $S$, namely
\begin{equation}\label{eq:distinguished_self_adjoin}
S_0\doteq\left.S^*\right|_{\ker\gamma_0},\quad S_1\doteq\left.S^*\right|_{\ker\gamma_1},
\end{equation}
which are transversal since $D(S)=D(S_0)\cap D(S_1)$ and $D(S^*)=D(S_0)+D(S_1)$, $+$ standing for the algebraic sum. The following proposition shows how to construct all other sought extensions, {\it cf.} \cite{Malamud} for a proof:

\begin{proposition}\label{Proposition: characterization of self adjoint extension via boundary triples}
Let $S$ be a closed, symmetric operator on $\mathsf{H}$. Then an associated boundary triple $(\mathsf{h},\gamma_0,\gamma_1)$ exists if and only if $S^*$ has equal deficiency indices. In addition, if $\Theta\colon D(\Theta)\subseteq\mathsf{h}\to\mathsf{h}$ is a closed and densely defined linear operator, then $S_\Theta\doteq S^*|_{\ker(\gamma_1-\Theta\gamma_0)}$ is a closed extension of $S$ with domain 
$$D(S_\Theta)\doteq\{f\in D(S^*)\;|\;\gamma_0(f)\in D(\Theta),\;\gamma_1(f)=\Theta\gamma_0(f)\}\,.$$
The map $\Theta\mapsto S_\Theta$ is one-to-one and  $S_\Theta^*=S_{\Theta^*}$, that is, it restricts to a one-to-one map from self-adjoint operators $\Theta$ to self-adjoint extensions of $S$.
\end{proposition}

\noindent Observe that, in this formulation, $S_1$ can be recovered by setting $\Theta=0$, while $S_0$ represents a kind of degenerate scenario, which has to be included by hand. The reason is due to the formulation of the proposition which we have chosen so to emphasize the connection with the heuristic notion of boundary conditions.

\begin{Example}
With reference to Example \ref{Ex:ordinary_boundary_triple}, self-adjoint extensions of $S$ are in one-to-one correspondence with surjective isometries $U\colon\mathcal{N}_-(S^*)\to\mathcal{N}_+(S^*)$, {\it cf.} \cite[Th. 5.37]{Moretti}, the domain of the extension being $D(S_U)=D(S)\oplus_S(\mathbb{I}-U)\mathcal{N}_-(S^*)$.
At the same time Proposition \ref{Proposition: characterization of self adjoint extension via boundary triples} guarantees that all self-adjoint extensions of $S$ are completely determined by the self-adjoint operators $\Theta:\mathsf{h}\to\mathsf{h}$.
Since the domain of the extension $S_\Theta$ is the collection of elements $f\in D(S^*)$ such that $\gamma_1(f)=\Theta\gamma_0(f)$, this is equivalent to requiring $(i\mathbb{I}-\Theta)\pi_+(f)=-(i\mathbb{I}+\Theta)V\circ\pi_-(f)$.
In other words $\pi_+(f)=\mathcal{C}(-\Theta)V\circ\pi_-(f)$, where $\mathcal{C}(-\Theta)$ stands for the Cayley transform of $\Theta$, hence a unitary operator on $\mathsf{h}$, {\it cf.} \cite[Th. 5.34]{Moretti}. By identifying $U=\mathcal{C}(-\Theta)V$ we have shown how to relate the two viewpoints. 
\end{Example}

\noindent To conclude this short digression, we emphasize how boundary triples also allow to obtain the spectral properties of the self-adjoint extensions of an Hermitian operator $S$ as above. To this end, first we need an auxiliary lemma, \cite{BL12}
\begin{lemma}\label{Lem:Decomposition_with_resolvent}
Let $S\colon D(S)\subseteq\mathsf{H}\to\mathsf{H}$ be a closed, symmetric operator and let $(\mathsf{h},\gamma_0,\gamma_1)$ be an associated boundary triple.
Let $S'$ be any self-adjoint extension of $S$, $S\subset S'\subset S^*$, and let $\lambda\in\rho(S')$, $\rho$ being the resolvent set.
Then 
\footnote{In this paper $V\dotplus W$ denotes the direct sum between the subspaces $V,W\subseteq\mathsf{H}$ ($V+W=\mathsf{H}$ and $V\cap W=\{0\}$) while $V\oplus W$ denotes the orthogonal direct sum ($V+W=\mathsf{H}$ and $V\perp W$).}
$$D(S^*)=D(S')\dotplus\mathcal{N}_\lambda(S^*),\quad\mathcal{N}_\lambda(S^*)\doteq\ker(S^*-\lambda\mathbb{I})\,.$$
\end{lemma}
\begin{proof}
Let $f\in D(S^*)$ and let $\tilde{f}=(S'-\lambda\mathbb{I})^{-1}((S^*-\lambda\mathbb{I})f)\in D(S')$, where $\lambda\in\rho(S')$.
For any $f^\prime\in D(S)$ it holds that
\begin{align*}
	((S'-\bar{\lambda}\mathbb{I})f^\prime,f-\tilde{f})_{\mathsf{H}}&=
	((S'-\bar{\lambda}\mathbb{I})f^\prime,f)_{\mathsf{H}} - ((S'-\bar{\lambda}\mathbb{I})f^\prime,(S'-\lambda\mathbb{I})^{-1}((S^*-\lambda\mathbb{I})f))_{\mathsf{H}}\\&=
	((S'-\bar{\lambda}\mathbb{I})f^\prime,f)_{\mathsf{H}} -(f^\prime,(S^*-\lambda\mathbb{I})f)_{\mathsf{H}}=0\,,
\end{align*}
where the last equality descends by using that $S^{**}=S$, being $S$ closed.
Hence $f-\tilde{f}\in\textrm{Ran}(S-\bar{\lambda}\mathbb{I})^\perp=\ker(S^*-\lambda\mathbb{I})=\mathcal{N}_\lambda(S^*)$.
We have found the sought decomposition.
To prove that it is unique, it suffices to observe that $D(S')\cap\mathcal{N}_\lambda(S^*)=\{0\}$ since $\lambda\in\rho(S')$ per assumption.
\end{proof}

\noindent We define the following auxiliary functions:

\begin{Definition}\label{Def:gamma_Weyl_functions}
Let $S\colon D(S)\subseteq\mathsf{H}\to\mathsf{H}$ be a closed, symmetric operator and let $(\mathsf{h},\gamma_0,\gamma_1)$ be an associated boundary triple.
Moreover, consider the self-adjoint extension $S_0$ of $S$ defined by $S_0\doteq S^*|_{\ker\gamma_0}$.
We call {\em $\gamma$-field} and {\em Weyl function} respectively the maps $\Gamma\colon\rho(S_0)\to D(S^*)$ and $M\colon\rho(S_0)\to\mathsf{h}$ such that
$$\Gamma(\lambda)\doteq\big[\gamma_0|_{\mathcal{N}_\lambda(S^*)}\big]^{-1},\quad M(\lambda)\doteq\gamma_1\circ\Gamma(\lambda)\,,$$
where $\rho(S_0)$ is the resolvent of $S_0$.
\end{Definition}

\noindent We have all the ingredients to state the following key theorem which specializes to the case at hand a statement proven in \cite{DM91}:

\begin{theorem}\label{Theorem: spectral properties via Weil function}
Let $S:D(S)\subseteq\mathsf{H}\to\mathsf{H}$ be a closed, symmetric operator and let $(\mathsf{h},\gamma_0,\gamma_1)$ be an associated boundary triple. Let $S_\Theta$ be a self-adjoint extension of $S$ determined by means of a self-adjoint operator $\Theta\colon D(\Theta)\subseteq\mathsf{h}\to\mathsf{h}$. Let $\rho$, $\sigma_p$ and $\sigma_c$ indicate respectively resolvent, point and continuous spectrum of an operator. Then for every $\lambda\in\rho(S_0)$, $S_0=\left.D(S^*)\right|_{\ker\gamma_0}$, it holds
\begin{enumerate}
	\item $\lambda\in\rho(S_\Theta)$ if and only if $0\in\rho(\Theta-M(\lambda))$, where $M$ is the Weyl function,
	\item $\lambda\in\sigma_i(S_\Theta)$, $i=p,c$ if and only if $0\in\sigma_i(\Theta-M(\lambda))$.
\end{enumerate}
\end{theorem} 

\noindent In other words, this theorem guarantees that the computation of the spectrum of $S_\Theta$, a self-adjoint extension of $S$, boils down to evaluating the spectra of $S_0$ and of $\Theta-M(\lambda)$.  

\subsection{Application to second-order elliptic differential operators}

In this section we apply the theory of boundary triples to the study of second-order, elliptic differential operators. We start our analysis from the distinguished, albeit special case of the Laplace-Beltrami operator, subsequently generalizing our construction. Hence, let $(M,g)$ be a Riemannian manifold with boundary and of bounded geometry as per Definition \ref{Def:bounded_geometry_boundary}.
On top of $M$ we consider the Laplace-Beltrami operator $\Delta_g$, which reads in a local chart $\Delta_g=-\nabla_ig^{ij}\nabla_j$, $\nabla$ being the Levi-Civita connection built out of $g$.
We wish to identify a boundary triple for $\Delta_g$, which is regarded as a densely defined operator $\Delta_g: H^2_0(M)\to L^2(M;d\mu_g)$, $H^2_0(M)$ being the closure of $\mathcal{D}(M)$ with respect to the $H^2(M)$-norm as defined in Proposition \ref{Prop:Sobolev_space_equivalence}, see also Remark \ref{Remark: Sobolev spaces on manifold with boundary}.
Standard arguments yield that $\Delta_g$ is a closed symmetric operator on $L^2(M)\equiv L^2(M;d\mu_g)$ whose adjoint $\Delta_g^*$ is defined on the so-called {\em maximal domain}
$$D_{\textrm{max}}(\Delta_g^*)\doteq\{f\in L^2(M)|\; \Delta_g(f)\in L^2(M)\}\,,\qquad\Delta_g^*f\doteq \Delta_gf\,.$$
Since $(M,g)$ is of bounded geometry, then $\Delta_g$ is uniformly elliptic and the maximal domain coincides with $H^2(M)$. Hence we shall identify $D_{\textrm{max}}(\Delta_g^*)=H^2(M)$.

To construct a boundary triple associated with $\Delta_g^*$, let $n$ be the outward-pointing unit normal of $\partial  M$ and let
$$\Gamma_0\colon H^2( M)\ni f\mapsto\Gamma f\in H^{3/2}(\partial M),\,\qquad\Gamma_1\colon H^2(M)\ni f\mapsto-\Gamma \nabla_n f\in H^{1/2}(\partial M)\,,$$
where $\Gamma$ is the trace map defined in Theorem \ref{Th:trace}. For $f_1,f_2\in H^2(M)=D_{\textrm{max}}(\Delta_g^*)$, the following Green's identity holds true:
\begin{gather}\label{Equation: boundary triple equation for Laplacian}
(\Delta_g^*f_1|\,f_2)_{L^2( M)}-(f_1|\,\Delta_g^*f_2)_{L^2( M)}=
(\Gamma_1f_1|\,\Gamma_0f_2)_{L^2(\partial  M)}-
(\Gamma_0f_1|\,\Gamma_1f_2)_{L^2(\partial M)}\,.
\end{gather}
Moreover, since the inner product $(\,|\,)_{L^2(\partial  M)}$ on $L^2(\partial M)\equiv L^2(\partial M;\iota^*_Md\mu_g)$, $\iota_M:\partial M\hookrightarrow M$, extends continuously to a pairing on $H^{-1/2}(\partial  M)\times H^{1/2}(\partial  M)$ as well as on $H^{-3/2}(\partial M)\times H^{3/2}(\partial M)$, there exist isomorphisms
$$j_\pm\colon H^{\pm 1/2}(\partial  M)\to L^2(\partial  M),\qquad \iota_\pm\colon H^{\pm 3/2}(\partial  M)\to L^2(\partial  M)\,,$$
such that, for all $(\psi,\phi)\in H^{1/2}(\partial M)\times H^{-1/2}(\partial M)$ and for all $(\widetilde{\psi},\widetilde{\phi})\in H^{3/2}(\partial M)\times H^{-3/2}(\partial M)$, 
\begin{align*}
(\psi,\phi)_{(1/2,-1/2)}=(j_+\psi|\,j_-\phi)_{L^2(\partial M)}\,,\quad
(\widetilde{\psi},\widetilde{\phi})_{(3/2,-3/2)}=(\iota_+\widetilde{\psi}|\,\iota_-\widetilde{\phi})_{L^2(\partial M)}\,,
\end{align*}
where $(\;,\;)_{(1/2,-1/2)}$ and $(\;,\;)_{(3/2,-3/2)}$ stand for the duality pairings between the associated Sobolev spaces.

Gathering all the above ingredients, the following result holds, being a generalization to the case of Riemannian manifolds with boundary and of bounded geometry of a classical result, {\it e.g.} \cite{Grubb68} :
\begin{proposition}\label{Proposition: boundary triple of the Laplacian}
Let $\Delta_g^*$ be the adjoint of the Laplace-Beltrami operator on a Riemannian manifold $(M,g)$ with boundary and of bounded geometry. Let 
\begin{gather}\label{Equation: Dirichlet boundary map for Laplacian}
\gamma_0\colon H^2(M)\ni f\mapsto \iota_+\Gamma_0f\in L^2(\partial  M)\,,\\
\label{Equation: Neumann boundary map for Laplacian}
\gamma_1\colon H^2(M)\ni f\mapsto j_+\Gamma_1f\in L^2(\partial  M)\,,
\end{gather}
Then $(L^2(\partial  M),\gamma_0,\gamma_1)$ is a boundary triple for $\Delta_g^*$.
Moreover the self adjoint extension $\Delta_{g,0}\doteq\Delta_g^*|_{\ker\gamma_0}$ coincides with the Dirichlet operator $\Delta_g^*|_{H^2( M)\cap H^1_0( M)}$.
\end{proposition}

\begin{proof}
	It suffices to observe that, in view of our assumption on the geometry of $(M,g)$, the trace map $\Gamma$ is both continuous and surjective, {\it c.f} \cite{GrosseSchneider}. Hence also the map $\gamma:H^2(M)\to L^2(\partial M)\times L^2(\partial M)$ is surjective by combining the properties of $\Gamma$ together with the fact that $\iota_+$ and $j_+$ are isomorphisms and $\nabla_n\colon H^s(M)\to H^{s-1}(M)$. Together with this data one can repeat the same proof as in \cite{Grubb68} in combination with a partition of unity argument. The last statement of this proposition descends from Theorem \ref{Th:trace} and, in particular from the kernel of $\Gamma$ acting on $H^s(M)$ being $H^1_0(M)\cap H^s(M)$ for all $s\geq 1$.
\end{proof}

\begin{remark}\label{Remark: non-local boundary conditions}

	Notice that Proposition \ref{Proposition: boundary triple of the Laplacian} ensures that the self-adjoint extensions of $\Delta_g$ are parametrized by \textit{all} densely defined self-adjoint operators $\Theta$ on $L^2(\partial M)$.
	Therefore, the class of boundary conditions $\ker(\gamma_1-\Theta\gamma_0)$ is more general then the one considered usually in the literature \cite{Hormander}.
\end{remark}

\begin{remark}\label{Remark: boundary triple for an elliptic operator}
We observe that, while the results of Proposition \ref{Proposition: boundary triple of the Laplacian} can be straightforwardly extended if one adds to $\Delta_g$ terms of order $0$ which are both smooth and bounded, dealing with a generic elliptic differential operator $A$ of order $2$ on $(M,g)$ requires a more careful investigation.

More specifically consider an operator $A$ which reads locally $A=-\nabla_ia^{ij}\nabla_j$, where both $a^{ij}$ and $\nabla_i a^{ij}\in C^\infty(M)\cap L^\infty(M)$ . One has individuated a
closed, symmetric operator $A\colon H^2_0( M)\to L^2(M)$ whose adjoint $A^*$ has a maximal domain
$$D(A^*)\doteq\{ f\in L^2( M)|\; Af\in L^2( M) \}\,,\qquad A^*f\doteq Af\,.$$
Since $(M,g)$ is of bounded geometry, $H^2( M)\subseteq D(A^*)$, though the inclusion is in general strict unless $A$ is uniformly elliptic. Since $H^2( M)$ is dense in $D(A^*)$, see \cite{Ibort:2013} and references therein, the maps
$$(\Gamma_0,\Gamma_1)\colon H^2(M)\ni f\mapsto (\Gamma f,-\Gamma\nabla_{a,n} f)\in H^{1/2}(\partial M)\times H^{3/2}(\partial M)\,,$$
where $\nabla_{a,n}f\doteq n_ia^{ij}\nabla_jf$ while $\Gamma$ is defined as in Theorem \eqref{Th:trace}, can be extended to linear continuous maps
$$(\widetilde{\Gamma}_0,\widetilde{\Gamma}_1)\colon D(A^*)\to H^{-1/2}(\partial  M)\times H^{-3/2}(\partial  M)\,.$$
In addition, if $f_1\in D(A^*)$ and $f_2\in H^2( M)$, it holds
\begin{align}\label{Equation: partial boundary triple equation for elliptic operator}
(A^*f_1|\,f_2)_{L^2(M)}-
(f_1|\;A^*f_2)_{L^2(M)} & =
(\widetilde{\Gamma}_1f_1|\;\Gamma_0f_2)_{L^2(\partial M)}-
(\widetilde{\Gamma}_0f_1|\;\Gamma_1f_2)_{L^2(\partial M)}
\\ \nonumber &=
( \iota_-\widetilde{\Gamma}_1f_1|\;\iota_+\Gamma_0f_2)_{-\frac 32,\frac 32}-
(j_-\widetilde{\Gamma}_0f_1|\;j_+\Gamma_1f_2)_{-\frac 12,\frac 12}\,,
\end{align}
where in the last equality we used the definition of $\iota_\pm$, $j_\pm$ while $(\,,\,)_{-\frac 32,\frac 32}$, $(\,,\,)_{-\frac 12,\frac 12}$ denote the pairing between dual spaces.
In order for the above data to define a boundary triple, let $A_0$ be the Dirichlet self-adjoint extension of $A$, defined as $A_0\doteq A^*|_{H^2(M)\cap H^1_0(M)}$.
By applying Lemma \ref{Lem:Decomposition_with_resolvent} for $S^*=A^*$ and for $S'=A_0$ it follows that, for an arbitrary but fixed $\lambda\in\rho(A_0)$, the maximal domain $D(A^*)$ decomposes as
\begin{gather}\label{Equation: algebraic decomposition of the maximal domain of an elliptic operator}
D(A^*)=D(A_0)\dotplus\ker(A^*-\lambda)\,.
\end{gather}
One can consider the maps
\begin{equation}\label{eq:boundary_triple_for_A*}
(\gamma_0,\gamma_1)\colon D(A^*)\ni f\mapsto (j_-\widetilde{\Gamma}_0f_\lambda,\iota_+\widetilde{\Gamma}_1f_0)=(j_-\widetilde{\Gamma}_0f_\lambda,\iota_+\Gamma_1f_0)\in L^2(\partial M)\times L^2(\partial M)\,,
\end{equation}
where $f_0\in D(A_0)$ is the ``Dirichlet part'' of $f=f_{0}+f_\lambda$ according to decomposition \eqref{Equation: algebraic decomposition of the maximal domain of an elliptic operator}, while $f_\lambda\in\ker(S^*-\lambda)$.
It can be shown that $(L^2(\partial  M),\gamma_0,\gamma_1)$ is a boundary triple for $A^*$ \cite{HSF12,Grubb71}, moreover, the self adjoint extension $S_0\doteq S^*|_{H^2(M)\cap H^1_0(M)}$ coincides with the self-adjoint extension $S^*|_{\ker\gamma_0}$, hence justifying the notation. We remark that the ambiguity in the choice of $\lambda\in\rho(S_0)$ reflects possible degeneracies of $A$.
\end{remark}

\section{Fundamental solutions on spacetimes with timelike boundary}\label{Sec:Fundamental_Solutions}

In this section, following Section \ref{Sec:Geometric_data} and Proposition \ref{prop:static_spacetime_with_boundary} in particular, we consider $(N,\widetilde{h})$ to be a {\em static Lorentzian spacetime} with timelike boundary.
Henceforth, for any $\Omega\subseteq N$, with $J^\pm(\Omega)$ we indicate the causal future ($+$) and the causal past ($-$) of $\Omega$, {\it cf.} \cite[Ch. 1]{BEE}.
On top of it we consider the D'Alembert wave operator which, in view of \eqref{eq:line_element}, reads locally 
\begin{gather}\label{eq:original_box}
\square_{\widetilde{h}}=-\widetilde{h}^{ij}\nabla_i\nabla_j=\beta^{-1}\partial_t^2-\frac 12 g^{ij}\partial_i(\ln\beta)\partial_j+\Delta_g\,,
\end{gather}
where $\Delta_g$ is the Laplace-Beltrami operator associated with $(M,g)$.
The solutions of the wave equation built out of \eqref{eq:original_box} are best characterized in terms of the fundamental solutions associated with $\square_{\widetilde{h}}$.

\begin{Definition}\label{Definition: advanced, retarded and causal propagator}
A continuous linear map $\widetilde{\mathsf{G}}^{\pm}\colon \mathcal{D}(\mathring{N})\to\mathcal{D}^\prime(\mathring{N})$ is called a retarded ($+$) or an advanced ($-$) Green's operator for $\square_{\widetilde{h}}$ on a static Lorentzian spacetime $(N,h)$ with timelike boundary if, for every $f\in \mathcal{D}(\mathring{N})$, $\mathring{N}\doteq N\setminus\partial N$,
\begin{align}\label{Equation: defining properties of advanced propagator}
\square_{\widetilde{h}}\circ\widetilde{\mathsf{G}}^\pm f=\widetilde{\mathsf{G}}^\pm\circ \square_{\widetilde{h}} f=f\,,\quad
\textrm{supp}(\widetilde{\mathsf{G}}^\pm f)\subseteq J^\mp(\textrm{supp}(f))\,.
\end{align}
The causal (or advanced-minus-retarded) propagator $\widetilde{\mathsf{G}}$ associated to the pair $(\widetilde{\mathsf{G}}^+,\widetilde{\mathsf{G}}^-)$ is the linear operator 
\begin{align}\label{Equation: causal propagator}
\widetilde{\mathsf{G}}\doteq\widetilde{\mathsf{G}}^--\widetilde{\mathsf{G}}^+\colon \mathcal{D}(\mathring{N})\to\mathcal{D}^\prime(\mathring{N})\,.
\end{align}
\end{Definition}

\begin{remark}\label{Remark_globally_hyperbolic}
In view of Proposition \ref{prop:static_spacetime_with_boundary}, under mild additional hypotheses, $(N,\widetilde{h})$ can be realized as a submanifold of a standard static, globally hyperbolic spacetime $(\widehat{N},\widehat{h})$.
Thereon one can consider the D'Alembert wave operator $\square_{\widehat{h}}$ and standard results in the theory of normally hyperbolic operators guarantee the existence of unique advanced/retarded fundamental solutions $\widehat{\mathsf{G}}^\pm$ with the same defining properties as in \eqref{Equation: causal propagator} and, in addition, $\widehat{\mathsf{G}}^\pm[\mathcal{D}(\widehat{N})]\subset C^\infty(\widehat{N})$.
As a consequence, whenever $\Omega\subset\mathring{N}$ endowed with $\widehat{h}|_\Omega$ identifies a globally hyperbolic submanifold of $N$ and thus also of $\widehat{N}$, it turns out that, calling $\widetilde{\mathsf{G}}^\pm_\Omega$ the fundamental solutions for the D'Alembert wave operator on $(\Omega,\widehat{h}|_\Omega)$, then $\widehat{\mathsf{G}}^\pm|_{\Omega}=\widetilde{\mathsf{G}}^\pm_\Omega$ and, whenever $\widetilde{\mathsf{G}}^\pm$ exist, $\widetilde{\mathsf{G}}^\pm|_{\Omega}=\widetilde{\mathsf{G}}^\pm_\Omega$.
The same line of reasoning cannot be used to compare $\widehat{\mathsf{G}}^\pm$ with $\widetilde{\mathsf{G}}^\pm$ since $N$ is not a globally hyperbolic submanifold of $\widehat{N}$.
We will be showing that the existence and uniqueness of $\widetilde{\mathsf{G}}^\pm$ is instead subordinated to the choice of `boundary conditions' at $\partial N$.
\end{remark}

\paragraph{Conformal reduction.}
In order to address the problem of the existence of $\widetilde{\mathsf{G}}^\pm$, it is convenient to consider on $(N,\widetilde{h})$ the conformally rescaled metric
\begin{align}\label{Equation: conformally rescaled metric}
h\doteq \beta^{-1}\widetilde{h}=-\textrm{d}t^2+\beta^{-1}g\,,
\end{align}
whose wave operator reads $\square_{h}=\partial^2_t+\Delta_{\beta^{-1}g}$.
The connection between these wave operators is well known, in particular
\begin{align}\label{Equation: conformal relation between wave operators}
\square_{\widetilde{h}}\circ\beta^{\frac{1-m}{4}}=
\beta^{-\frac{3+m}{4}}\bigg[\square_{h}
+\frac{1-m}{2}\beta^{-\frac 12}\square_h(\beta^{\frac 12})-\frac{(1-m)(m-3)}{4}\beta\,h^{ij}\nabla_i\beta\nabla_j\beta
\bigg]\doteq
\beta^{-\frac{3+m}{4}}\bigg[\partial_t^2+A\bigg]\,,
\end{align}
where $\dim N=m+1$ while $A$ is an elliptic operator as per remark \ref{Remark: boundary triple for an elliptic operator}.
In the following we will give a construction for the advanced and retarded Green's operators $\mathsf{G}^\pm$ for $\partial_t^2+A$ rather that for $\square_{\widetilde{h}}$. Its counterparts, namely $\widetilde{\mathsf{G}}^\pm$, are related to the former ones via the identity
\begin{align}\label{Equation: conformal relation between propagators}
\widetilde{\mathsf{G}}^\pm=
\beta^{\frac{1-m}{4}}\circ\mathsf{G}^\pm\circ\beta^{\frac{3+m}{4}}\,.
\end{align}

In the following our results will depend on the regularity properties of $\beta$ and in particular we shall select it in such a way that the spacetime $(N,h)$ meets the conditions of Proposition \ref{prop:static_spacetime_with_boundary}, where $h=-\textrm{d}t^2+\beta^{-1}g$.
This entails in particular that the operator
\begin{align}\label{Equation: elliptic operator with beta dependence}
A\doteq
\Delta_{\beta^{-1}g}+
\frac{1-m}{2}\beta^{-\frac 12}\Delta_{\beta^{-1}g}(\beta^{\frac 12})-\frac{(1-m)(m-3)}{4}\,g^{ij}\nabla_i(\beta)\nabla_j(\beta)\,,
\end{align}
is uniformly elliptic.

In order to tackle the problem of the existence of $\mathsf{G}^\pm$ as per Definition \ref{Definition: advanced, retarded and causal propagator},
one can use Schwartz kernel theorem, rewriting the underlying problem in terms of two distributions $\mathcal{G}^\pm\in\mathcal{D}^\prime(\mathring{N}\times\mathring{N})$ such that, besides the condition on the support stated in \eqref{Equation: defining properties of advanced propagator}, 
\begin{equation}\label{eq:auxiliary_fundamental}
[(\partial_t^2+A)\otimes\mathbb{I}]\mathcal{G}^\pm=[\mathbb{I}\otimes(\partial_t^2+A)]\mathcal{G}^\pm=\delta_{\textrm{diag}(\mathring{N})}\,.
\end{equation}
Here $\delta_{\textrm{diag}(\mathring{N})}$ being the Dirac-delta distribution on the diagonal $\textrm{diag}(\mathring{N})\stackrel{\iota}{\hookrightarrow}\mathring{N}\times\mathring{N}$ of $\mathring{N}\times\mathring{N}$ such that, for every $f\in\mathcal{D}(\mathring{N}\times\mathring{N})$, $\delta_{\textrm{diag}(\mathring{N})}(f)=\int\limits_N \iota^*f \textrm{d}\mu_h$, where $\textrm{d}\mu_h$ is the metric induced volume form on $N$.
Since $(N,h)$ is per assumption a standard static spacetime, realizing it as the warped product $M\times_{\beta}\mathbb{R}$, we can make the ansatz $\mathcal{G}^-=\theta(t-t^\prime)\mathcal{G}$ and $\mathcal{G}^+=-\theta(t^\prime-t)\mathcal{G}$, where $\theta$ stands for the Heaviside distribution on $\mathbb{R}\times\mathbb{R}$ which is implicitly extended to $\mathring{N}\times\mathring{N}$.
From \eqref{eq:auxiliary_fundamental} it turns out that $\mathcal{G}\in\mathcal{D}^\prime(\mathring{N}\times\mathring{N})$ is the distribution associated to $\mathsf{G}$ via Schwartz kernel theorem and it satisfies 
\begin{equation}\label{eq:PDE_for_G}
\left\{\begin{array}{l}
[(\partial_t^2+A)\otimes\mathbb{I}]\mathcal{G}=[\mathbb{I}\otimes(\partial_t^2+A)]\mathcal{G}=0\\
\mathcal{G}|_{t=t^\prime}=0,\quad\partial_t\mathcal{G}|_{t=t^\prime}=-\partial_{t^\prime}\mathcal{G}|_{t=t^\prime}=\delta_{\textrm{diag}(\mathring{M})}
\end{array}\right. ,
\end{equation}
where $|_{t=t^\prime}$ indicates the pull-back of a distribution on $\mathring{N}\times\mathring{N}$ to $(\{t\}\times\mathring{M})\times (\{t\}\times\mathring{M})$ while $\delta_{\textrm{diag}(\mathring{M})}$ stands for the Dirac delta distribution on the diagonal of $\mathring{M}\times\mathring{M}$.
Observe that the product between the Heaviside distribution and $\mathcal{G}$ which defines $\mathcal{G}^\pm$ is well-posed as a consequence of \cite[Th. 8.2.10]{Hormander}.

\begin{remark}\label{Remark: discussion of the globally hyperbolic ultrastatic case}
From the previous discussion, it emerges that the existence of the Green's operators $\mathsf{G}^\pm$ can be reduced to finding first a solution for \eqref{eq:PDE_for_G} and proving, subsequently, that the support properties in \eqref{Equation: defining properties of advanced propagator} are verified.
In order to motivate the strategy that we shall follow to address these two questions for static Lorentzian spacetimes with timelike boundary, let us consider a special instance of the scenario of Remark \ref{Remark_globally_hyperbolic}.
$(N,h)$ is realized as a submanifold of a globally hyperbolic, standard static spacetime $(\widehat{N},\widehat{h})$ where the warping factor $\widehat{\beta}$ of Proposition \ref{prop:static_spacetime_with_boundary} is assumed to be equal to $1$.
Hence $\widehat{N}$ is isometric to $\widehat{M}\times\mathbb{R}$ where $(\widehat{M},\widehat{g})$ is a complete Riemannian manifold.
The d'Alembert wave operator reads $\square_{\widehat{h}}=\partial^2_t+\Delta_{\widehat{g}}$.
In this framework the Laplace-Beltrami operator $\Delta_{\widehat{g}}$ is known to be an essentially self-adjoint operator on $L^2(\widehat{M})\equiv L^2(\widehat{M};\textrm{d}\mu_{\widehat{g}})$, $\textrm{d}\mu_{\widehat{g}}$ being the metric induced volume form.
Indicating with a slight abuse of notation still with $\Delta_{\widehat{g}}$ the unique self-adjoint extension, for every $f\in\mathcal{D}(\widehat{M}\times\mathbb{R})$ it holds
\begin{align*}
\widehat{\mathcal{G}}(f_1,f_2)=
\int_{\mathbb{R}^2}\,\bigg(f_1(t)\bigg|\,\Delta_{\widehat{g}}^{-\frac{1}{2}}\sin\big[\Delta_{\widehat{g}}^{\frac{1}{2}}(t-t^\prime)\big]f_2(t^\prime)\bigg)\textrm{d}t\textrm{d}t'\,,
\end{align*}
where $f(t)\in H^2(M)$ denotes the evaluation of $f$, regarded as an element of $C_{\textrm{c}}^\infty(\mathbb{R},H^2(\widehat{M}))$.
Here $\Delta_{\widehat{g}}^{-\frac{1}{2}}\sin\big[\Delta_{\widehat{g}}^{\frac{1}{2}}(t-t^\prime)\big]$ in a densely defined self-adjoint operator on $L^2(M)$ which is defined exploiting the functional calculus for $\widehat{\Delta}_g$ \cite[Chap. 7]{ReedSimon}.
\end{remark}

From the example outlined in this last remark, it turns out that a solution of \eqref{eq:PDE_for_G} can be identified by exploiting the boundary triples introduced in the previous section.
Indeed, by fixing a boundary condition $\Theta$ we are reduced to dealing with the operator $\partial_t^2+A_\Theta$ for which we can exploit the spectral calculus associated with $A_\Theta$ as in the previous remark \ref{Remark: discussion of the globally hyperbolic ultrastatic case}.
This procedure will define a causal propagator for $\partial_t^2+A$ which can also be regarded as the unique causal propagator for $\partial_t^2+A_\Theta$.
The following theorem translates this paradigm:

\begin{theorem}\label{Theorem: construction of the advanced and retarded propagators}
Let $(N,h)$ be a static Lorentzian spacetime with timelike boundary as per Proposition \ref{prop:static_spacetime_with_boundary}.
Let $(\gamma_0,\gamma_1,L^2(\partial M))$ be the boundary triple as in \eqref{eq:boundary_triple_for_A*} associated with the elliptic operator $A^*$ defined in \eqref{Equation: elliptic operator with beta dependence} and let $\Theta$ be a densely defined self-adjoint operator on $L^2(\partial M)$.
Let $A_\Theta$ be the self-adjoint extension of $A$ defined as per Proposition \ref{Proposition: characterization of self adjoint extension via boundary triples} by $A_\Theta\doteq A^*|_{D(A_\Theta)}$, where $D(A_\Theta)\doteq\ker(\gamma_1-\Theta\gamma_0)$.
Furthermore, let assume that the spectrum of $A_{\Theta}$ is bounded from below.\\
Then the advanced and retarded Green's operators $\mathsf{G}^\pm_\Theta$ associated to the wave operator
$\partial_t^2+A_\Theta$ exist and they are unique.
They are completely determined in terms of the bidistributions $\mathcal{G}^-_\Theta=\theta(t-t^\prime)\mathcal{G}_\Theta$ and $\mathcal{G}^+_\Theta=-\theta(t^\prime-t)\mathcal{G}_\Theta$, where $\mathcal{G}_\Theta\in\mathcal{D}^\prime(\mathring{N}\times\mathring{N})$ is such that, for all $f\in \mathcal{D}(\mathring{N})$
\begin{align}\label{Equation: construction of the causal propagator}
\mathcal{G}_\Theta(f_1,f_2)\doteq\int_{\mathbb{R}^2}\textrm{d}t\textrm{d}t'\,\bigg( f_1(t)\bigg| A_\Theta^{-\frac{1}{2}}\sin\big[A_\Theta^{\frac{1}{2}}(t-t^\prime)\big]f_2(t^\prime)\bigg),
\end{align}
where $f(t)\in H^2(M)$ denotes the evaluation of $f$, regarded as an element of $C_{\textrm{c}}^\infty(\mathbb{R},H^\infty(M))$ and $A_\Theta^{-\frac{1}{2}}\sin\big[A_\Theta^{\frac{1}{2}}(t-t^\prime)]$ is defined exploiting the functional calculus for $A_{\Theta}$.
Moreover it holds that 
$$\mathsf{G}^\pm_\Theta\colon\mathcal{D}(\mathring{N})\to C^\infty(\mathbb{R},H^\infty_\Theta(M))\,,$$
where $H^\infty_\Theta(M)\doteq\bigcap_{k\geq 0}D(A_\Theta^k)$. 
In particular,
\begin{align}\label{Equation: boundary condition for advanced and retarded propagators}
\gamma_1\big(\mathsf{G}^\pm_\Theta f\big)=\Theta\gamma_0\big(\mathsf{G}^\pm_\Theta f\big)
\qquad\forall f\in C^\infty_0(\mathring{N})\,.
\end{align}
\end{theorem}

\begin{proof}
As starting point, we observe that, since per assumption $A_{\Theta}$ is bounded from below,
the function $\sigma(A_\Theta)\ni\lambda\mapsto\lambda^{-1/2}\sin(\sqrt{\lambda}\tau)$ is smooth and bounded for all $\tau\in\mathbb{R}$.
Hence, for any $f\in C^\infty_0(\mathring{N})$ \eqref{Equation: construction of the causal propagator} identifies an element $\mathsf{G}_\Theta f\in C^\infty(\mathbb{R},D(A^{1/2}_{\Theta}))\subset\mathcal{D}^\prime(\mathring{N})$. Moreover, for all $k\in\mathbb{N}\cup\{0\}$ and for all $t\in\mathbb{R}$, we have
$$(1+A_\Theta)^k[\mathsf{G}_\Theta f](t)=\mathsf{G}_\Theta[(1+A_\Theta)^kf](t)=\mathsf{G}_\Theta[(1+A)^kf](t)\,,$$
which identifies an element of $L^2(M)$ since  $(1+A)^kf\in \mathcal{D}(\mathring{N})$.
Therefore $[\mathsf{G}_\Theta f](t)$ lies in $H^\infty_\Theta(M)$ and equation \eqref{Equation: boundary condition for advanced and retarded propagators} follows.
Notice that, since the hypothesis on $\beta$ guarantees that $A$ is a uniformly elliptic operator, $H^\infty_\Theta(M)\subset H^\infty(M)$.

The operator $\mathsf{G}_\Theta$ identifies unambiguously $\mathcal{G}_\Theta\in\mathcal{D}^\prime(\mathring{N}\times\mathring{N})$ satisfying per construction \eqref{eq:PDE_for_G}. In order to claim that $\mathcal{G}^+\doteq\theta(t-t^\prime)\mathcal{G}$ and $\mathcal{G}^-=-\theta(t^\prime-t)\mathcal{G}$ are the sought fundamental solutions we need to prove that $\textrm{supp}(\mathcal{G}^\pm(f))\subseteq J^\mp(\textrm{supp}(f))$. To this end it suffices to show that $\textrm{supp}(\mathcal{G}(f))\subseteq J(\textrm{supp}(f))$, $J(\textrm{supp}(f))\doteq J^-(\textrm{supp}(f))\cup J^+(\textrm{supp}(f))$.

This property is in turn a consequence of an energy estimate.
Let $u(t)\doteq[\mathcal{G}_\Theta(f)](t)$ and let $K\subseteq  M\simeq \{0\}\times M$ be a compact set.
We prove that, if $u(0)|_K=\dot{u}(0)|_K=0$, then $u$ vanishes on $U_K\doteq N\setminus J( M\setminus K)$.
As a consequence $u(t)$ vanishes on $K_t\doteq M_t\cap U_K$ for all $t\in\mathbb{R}$, where $ M_t\doteq\{t\}\times M$.
To this end we introduce the positive energy functional
$$E_K[u](t)\doteq\frac 12\big[C_\infty\|u(t)\|^2_{L^2(K_t)}+\|\dot{u}(t)\|^2_{L^2(K_t)}+\|\nabla_{\beta^{-1}g} u(t)\|^2_{L^2(K_t)}\big]\,,$$
where $\dot{u}(t)$ stands for $\frac{\textrm{d}}{\textrm{d}t}u(t)$ and where square integrability is defined still with respect to the metric-induced volume form, here left implicit.
Moreover we set $C_\infty\doteq\sup_{U_K}|C|>0$ where $C:=A-\Delta_{\beta^{-1}g}$, see equation \eqref{Equation: elliptic operator with beta dependence}.
A direct computation shows that
\begin{align*}
\frac{\textrm{d}}{\textrm{d}t}E_K[u](t)&=
C_\infty(\dot{u}(t)|u(t))-(\dot{u}(t)|[\Delta_{\beta^{-1}g}+C] u(t))+
(\nabla_{\beta^{-1}g}\dot{u}(t)|\nabla_{\beta^{-1}g} u(t))\\&-
\frac 12\big[C_\infty
\|\gamma_{0,K_t}u(t)\|^2_{L^2(\partial K_t)}+
\|\gamma_{0,K_t}\dot{u}(t)\|^2_{L^2(\partial K_t)}
+\|\gamma_{1,K_t}u(t)\|^2_{L^2(\partial K_t)}\big]\,,
\end{align*}
where $\gamma_{0,K_t}u(t)\doteq u(t)|_{\partial K_t}$, $\gamma_{1,K_t}u(t)\doteq\nabla_{\beta^{-1}g} u(t)|_{\partial K_t}$ are the traces of $u(t)$, $\nabla_{\beta^{-1}g} u(t)$ on $\partial K_t$, which are well-defined in view of the regularity of $u(t)$.
The divergence theorem yields
\begin{align*}
(\nabla_{\beta^{-1}g}\dot{u}(t)|\nabla_{\beta^{-1}g} u(t))&=
(\gamma_{0,K_t}\dot{u}(t)|\gamma_{1,K_t}u(t))_{L^2(\partial K_t)}+
(\dot{u}(t)|\Delta_{\beta^{-1}g} u(t))\\&\leq
\frac 12\big[
\|\gamma_{0,K_t}\dot{u}(t)\|^2_{L^2(\partial K_t)}+
\|\gamma_{1,K_t} u(t)\|^2_{L^2(\partial K_t)}\big]+
(\dot{u}(t)|\Delta_{\beta^{-1}g} u(t))\,.
\end{align*}
From this estimate we obtain a bound for the time derivative of $E_K[u](t)$, namely
\begin{align*}
\frac{\textrm{d}}{\textrm{d}t}E_K[u](t)&\leq
C_\infty(\dot{u}(t)|u(t))-(\dot{u}(t)|Cu(t)) 
\leq b E_K[u](t)\,,
\end{align*}
where $b>0$.
By Gr\"onwall's inequality it follows that $E_K[u](t)\leq e^{bt}E_K[u](0)$ which vanishes because of our hypothesis on $u$.
Since $E_K[u](t)$ is positive it follows that $u(t)=0$ vanishes on $K_t$ for all $t\in\mathbb{R}$. This proves that $\mathcal{G}$ enjoys the desired support property. In addition the energy estimate also entails the uniqueness of $\mathcal{G}$ and in turn of $\mathcal{G}^\pm$. Suppose a second pair of retarded and advanced fundamental solutions exist, say $\mathcal{G}^\pm_2$. Then, calling $\mathcal{G}_2=\mathcal{G}^-_2-\mathcal{G}^+_2$, it would hold that, for every $f\in\mathcal{D}(N)$, $\mathcal{G}(f,\cdot)-\mathcal{G}_2(f,\cdot)$ is a solution of the D'Alembert wave equation with vanishing initial data. The energy estimate entails that this must be the zero solution, from which uniqueness descends.
\end{proof}

\begin{remark}
One may wonder for which $\Theta$ the self-adjoint extension $A_\Theta$ is bounded from below.
To the best of our knowledge this is an open problem.
In \cite{Grubb70} it was shown that, if the Dirichlet extension $A_0=A^*|_{\ker\gamma_0}$ has a finite lower bound $m(A_0)$ then $A_\Theta$ has a finite lower bound $m(A_\Theta)$ whenever that $\Theta$ has a lower bound such that $m(\Theta)>-m(A_0)$.
Moreover, in this latter case it holds $m(A_\Theta)>m(\Theta)m(A_0)[m(\Theta)+m(A_0)]^{-1}$.
\end{remark}

\begin{remark}
	Observe that, since the operator $A$ is per assumption uniformly elliptic, $H^\infty_\Theta(M)\subset H^\infty(M)$ where $H^\infty(M)\doteq\bigcap_{k\geq 0} H^k(M)$. In view of Remark \ref{Rem:quotient_Spaces}, we can identify this space as the quotient between $H^\infty(\widehat{M})$ and the collection of those elements $v\in H^\infty(\widehat{M})$ such that $v|_M=0$. Since $\widehat{M}$ is a metric complete Riemannian manifold, being of bounded geometry, it turns out that each element $v\in H^\infty(\widehat{M})$ admits a representative in $C^\infty(\widehat{M})$, see \cite[Ch. 3]{Hebey}. Consequently every $u\in H^\infty_\Theta(M)$ admits a representative lying in $C^\infty(M)$.
\end{remark}

\begin{remark}
	The previous analysis proves existence and uniqueness of the fundamental solutions for a wide class of boundary conditions.
	Yet it is important to mention that using the theory of boundary triples is not the only possible path which can be followed.
	In \cite{Posilicano-18} self-adjoint extensions of the restriction of a given self-adjoint operator are classified via an efficient parametrization of their resolvent.
	In \cite{Ibort:2013, Ibort:2014sua}, a class of self-adjoint extensions of the Laplace-Beltrami operator on a smooth Riemannian manifold $M$ with smooth compact boundary $\partial M$ has been studied in the framework of quadratic forms.
	While the extension to a non-compact boundary has not been investigated yet, this method highlights that, among the collection of boundary conditions, a distinguished class is represented in terms of their interplay with the unitary representations of the isometry group of $\partial M$ acting on $L^2(\partial M)$. 
\end{remark}

\begin{remark}
	We expect that a result similar in spirit to Theorem \ref{Theorem: construction of the advanced and retarded propagators} can be derived also in the framework of \cite{DerezinskiSiemssen17}, up to a suitable modification of the geometrical setting considered therein.
	This would allow to consider the wave operator with the insertion of an electromagnetic potential and with possibly low regularity of both the metric and the electromagnetic potential in the sense of \cite{DerezinskiSiemssen17}.
	We are currently investigating this topic.
\end{remark}

To conclude the analysis of the structural properties of the fundamental solutions, we extend to the case at hand a result which, in the case of globally hyperbolic spacetimes, was proven in \cite[Th. 3.4.7]{BGP}. The goal is to prove that the advanced and retarded Green operators allow for a complete characterization of the space of smooth solutions of the D'Alembert wave equation.
In view of equations \eqref{Equation: conformal relation between wave operators} and \eqref{Equation: conformal relation between propagators} we can still work with the operator $\partial^2_t+A$. To this end, first of all we need to enlarge the domain of definition of the operators $\mathsf{G}^\pm_\Theta$ defined in Theorem \ref{Theorem: construction of the advanced and retarded propagators}.
\begin{Definition}
Under the same hypotheses of Theorem \ref{Theorem: construction of the advanced and retarded propagators}, let $H^\infty_\Theta(M)\doteq\bigcap_{k\geq 1}D(A_\Theta^k)$. Hence, motivated by the nomenclature in \cite{Baer15} we call:
\begin{enumerate}
	\item $C^\infty_{\textrm{fc}}(\mathbb{R},H^\infty_\Theta(M))$ the space of {\em future compact}, $H^\infty_\Theta$-valued smooth functions, that is the collection of $f\in C^\infty(\mathbb{R},H^\infty_\Theta(M))$ for which $J^+_{\mathbb{R}\times_{\beta} M}(x)\cap\textrm{supp}(f)$ is compact or empty for all $x\in\mathbb{R}\times_{\beta} M$, where $J^+_{\mathbb{R}\times_{\beta} M}$ stands for the causal future in $\mathbb{R}\times_{\beta} M$ endowed with the metric \eqref{Equation: conformally rescaled metric},
	\item $C^\infty_{\textrm{pc}}(\mathbb{R},H^\infty_\Theta(M))$ the space of {\em past compact} $H^\infty_\Theta(M)$-valued smooth functions, that is the collection of $f\in C^\infty(\mathbb{R},H^\infty_\Theta(M))$ for which $J^-_{\mathbb{R}\times_{\beta} M}(x)\cap\textrm{supp}(f)$ is compact or empty for all $x\in\mathbb{R}\times_{\beta} M$, where $J^-_{\mathbb{R}\times_{\beta} M}$ stands for the causal past in $\mathbb{R}\times_{\beta} M$ endowed with the metric \eqref{Equation: conformally rescaled metric},
	\item $C^\infty_{\textrm{tc}}(\mathbb{R},H^\infty_\Theta(M))\doteq C^\infty_{\textrm{pc}}(\mathbb{R},H^\infty_\Theta(M))\cap C^\infty_{\textrm{fc}}(\mathbb{R},H^\infty_\Theta(M))$ the collection of {\em timelike compact} $H^\infty_\Theta(M)$-valued smooth functions.
\end{enumerate}
\end{Definition}
\noindent By combining Theorem \ref{Theorem: construction of the advanced and retarded propagators} and the results of \cite{Baer15}, it descends that $\mathsf{G}^\pm_\Theta$ extend respectively to linear maps
\begin{align*}
\mathsf{G}^+_\Theta\colon C^\infty_{\textrm{fc}}(\mathbb{R},H^\infty_\Theta(M))\to C^\infty(\mathbb{R},H^\infty_\Theta(M))\,,\qquad
\mathsf{G}^-_\Theta\colon C^\infty_{\textrm{pc}}(\mathbb{R},H^\infty_\Theta(M))\to C^\infty(\mathbb{R},H^\infty_\Theta(M))\,,
\end{align*}
which preserve the properties $\square_\Theta\circ\mathsf{G}^\pm_\Theta=\mathsf{G}^\pm_\Theta\circ\square_\Theta=\textrm{id}$ and $\textrm{supp}(\mathsf{G}^\pm f)\subseteq J^\mp(\textrm{supp}(f))$.

\begin{proposition}\label{Prop:exact_sequence}
Under the hypotheses of Theorem \ref{Theorem: construction of the advanced and retarded propagators}, let $\square_\Theta\doteq\partial_t^2+A_\Theta$.
Then the following is an exact sequence of linear maps:
\begin{align}\label{Equation: short exact sequence for spacetimes with timelike boundary}
0\to C^\infty_{\textrm{tc}}(\mathbb{R},H^\infty_\Theta(M))
\stackrel{\square_\Theta}{\longrightarrow}
C^\infty_{\textrm{tc}}(\mathbb{R},H^\infty_\Theta(M))
\stackrel{\mathsf{G}_\Theta}{\longrightarrow}
C^\infty(\mathbb{R},H^\infty_\Theta(M))
\stackrel{\square_\Theta}{\longrightarrow}
C^\infty(\mathbb{R},H^\infty_\Theta(M))\to 0\,.
\end{align}
\end{proposition} 

\begin{proof}
The proof follows the same steps of \cite[Th. 3.4.7]{BGP}.
Let $f\in C^\infty_{\textrm{tc}}(\mathbb{R},H^\infty_\Theta(M))$ be such that $\square_\Theta f=0$. Hence $f=\mathsf{G}^+_\Theta\square_\Theta f =0$, which shows the injectivity of the first arrow ruled by $\square_\Theta$. 

For all $f\in C^\infty_{\textrm{tc}}(\mathbb{R},H^\infty_\Theta(M))$ it holds from Definition \ref{Definition: advanced, retarded and causal propagator} 
that $\square_\Theta\mathsf{G}_\Theta f=\mathsf{G}_\Theta\square_\Theta f=0$. Therefore $\square_\Theta C^\infty_{\textrm{tc}}(\mathbb{R},H^\infty_\Theta(M))\subseteq\ker\mathsf{G}_\Theta$. Let now $f\in C^\infty_{\textrm{tc}}(\mathbb{R},H^\infty_\Theta(M))$ be such that $\mathsf{G}_\Theta f=0$. Then $\psi=\mathsf{G}^+_\Theta f =-\mathsf{G}^-f$ belongs to $C^\infty_{\textrm{tc}}(\mathbb{R},H^\infty_\Theta(M))$ and $\square_\Theta\psi=f$, which entails $f\in\square_\Theta C^\infty_\Theta(\mathbb{R},H^\infty_\Theta(M))$. Hence $\textrm{Im}(\square_\Theta)=\ker\mathsf{G}_\Theta$ which prove the exactness of the second and of third arrow. We now show that
$\mathsf{G}_\Theta[C^\infty_{\textrm{tc}}(\mathbb{R},H^\infty_\Theta(M))]=\ker[\square_\Theta|_{C^\infty(\mathbb{R},H^\infty_\Theta(M))}]$.
The inclusion $\subseteq$ follows from $\square_\Theta\mathsf{G}_\Theta C^\infty_{\textrm{tc}}(\mathbb{R},H^\infty_\Theta(M))=0$.
Let $f\in C^\infty(\mathbb{R},H^\infty_\Theta(M))$ be such that $\square_\Theta f=0$.
Let us consider its (non-unique) decomposition as $f=f_++f_-$ with $f_\pm\in C^\infty(\mathbb{R},H^\infty_\Theta(M))$ where $f_+$ (resp. $f_-$) is future (resp. past) compact. It follows that $\square_\Theta f_-\in C^\infty_{\textrm{pc}}(\mathbb{R},H^\infty_\Theta(M))\subset C^\infty(\mathbb{R},H^\infty_\Theta(M))$ and from the support properties of $\mathsf{G}^\pm_\Theta$ combined with $\square_\Theta f_+=-\square_\Theta f_-$,
$$\mathsf{G}_\Theta\square_\Theta f_-=\mathsf{G}^-_\Theta\square_\Theta f_-+\mathsf{G}^+_\Theta\square_\Theta f_+ = f_++f_-=f\,,$$
from which it descends the exactness of the third and forth arrows.
Finally, the last arrow entails instead that, for every $h\in C^\infty(\mathbb{R},H^\infty_\Theta(M))$ there exists $f\in C^\infty(\mathbb{R},H^\infty_\Theta(M))$ such that $\Box_\Theta f=h$. Similarly to the proof of the previous points, it suffices to split $h=h_+ +h_-$ where $h_+\in C^\infty_{\textrm{fc}}(\mathbb{R},H^\infty_\Theta(M))$ while $h_-\in C^\infty_{\textrm{pc}}(\mathbb{R},H^\infty_\Theta(M))$. As a consequence of the properties of $\mathsf{G}^\pm_\Theta$, the sought function is $f=\mathsf{G}^+h_++\mathsf{G}^- h_-$.
\end{proof}

\begin{remark}
	Observe that another, equivalent way of reading the sequence \eqref{Equation: short exact sequence for spacetimes with timelike boundary} of Proposition \ref{Prop:exact_sequence} is that the causal propagator $\mathsf{G}_\Theta\doteq\mathsf{G}^-_\Theta-\mathsf{G}^+_\Theta$ induces an isomorphism between the quotient of vector spaces $\frac{C^\infty_{\textrm{tc}}(\mathbb{R},H^\infty_\Theta(M))}{\square_\Theta[C^\infty_{\textrm{tc}}(\mathbb{R},H^\infty_\Theta(M))]}$ and 
	$\mathcal{S}_\Theta(\mathbb{R}\times_{\beta} N)\doteq\{u\in C^\infty(\mathbb{R}; H^\infty_\Theta(M))\;|\;\square_\Theta u = 0\}.$
\end{remark}

\begin{Example}\label{Example: propagators for half space}
To better illustrate the above analysis we discuss in detail an explicit example often used in the literature. Most notably we consider a static Lorentzian spacetime with boundary $N=M\times\mathbb{R}$ with warping factor $\beta=1$. The underlying manifold $(M,g)$ is assumed to be $\mathbb{R}_+\times\mathbb{R}^n$, $n\geq 0$, endowed with the standard Euclidean metric. Let $(\gamma_0,\gamma_1,L^2(\mathbb{R}^n))$ be the boundary triple as per Proposition \ref{Proposition: boundary triple of the Laplacian} for $n>1$ while, for $n=0$ the boundary Hilbert space is $\mathbb{C}$.
The Laplace-Beltrami operator associated to $(M,g)$ reads in Cartesian coordinates
$$\Delta_g=-\partial_{x_1}^2-\sum\limits_{i=2}^{n+1}\partial_{x_i}^2\,,$$
and we indicate with $\Delta_\Theta$ the self-adjoint extension induced by $\Theta$ according to Theorem \ref{Theorem: spectral properties via Weil function}. If $n=0$ $\Delta$ becomes simply $-\partial_{x_1}^2$. This is a special case, for which the following discussion is not necessary. In order to apply Theorem \ref{Theorem: construction of the advanced and retarded propagators}, the first step consists of checking that $\Delta_{\Theta}$ is bounded from below.
Taking the Fourier transform $\mathcal{F}_{\mathbb{R}^n}$ in all variables barring $x_1$ and using the notational shortcut $k^2\doteq\sum\limits_{i=2}^{n+1}k^2_i$ this question reduces to the study of the spectral property of 
$$\widetilde{\Delta}\doteq(1\otimes\mathcal{F}_{\mathbb{R}^n})^{-1}\circ\Delta_g\circ(1\otimes\mathcal{F}_{\mathbb{R}^n})=-\partial_{x_1}^2+k^2\,.$$
defined on the domain (the tilde symbol indicating that the domains refer to the operator $\widetilde{\Delta}$)
\begin{align*}
\widetilde{H}^2_0(\mathbb{R}_+\times\mathbb{R}^n)=
\{\psi\in L^2(\mathbb{R}_+\times\mathbb{R}^n)\;|\; \partial_{x_1}^2\psi,k^2\psi\in L^2(\mathbb{R}_+\times\mathbb{R}^n)\,,\quad\gamma_0(\psi)=\gamma_1(\psi)=0\}\,.
\end{align*}
The adjoint is defined on
\begin{align*}
\widetilde{H}^2(\mathbb{R}_+\times\mathbb{R}^n)=
\{\psi\in L^2(\mathbb{R}_+\times\mathbb{R}^n)\;|\; \partial_{x_1}^2\psi,k^2\psi\in L^2(\mathbb{R}_+\times\mathbb{R}^n)\}\,,\qquad
[\widetilde{\Delta}_h]^*\psi=\widetilde{\Delta}_h\psi.
\end{align*}

\noindent
The deficiency indexes of $\widetilde{\Delta}$ are equal, in particular $\ker(\widetilde{\Delta}\pm i)\simeq L^2(\mathbb{R}^n)$.
According to Proposition \ref{Proposition: characterization of self adjoint extension via boundary triples} any self-adjoint extension of $\widetilde\Delta$ can be obtained from a self-adjoint operator $\widetilde\Theta$ on $L^2(\mathbb{R})$ as
\begin{align*}
D(\widetilde{\Delta}_{\widetilde\Theta})\doteq\{\psi\in \widetilde{H}^2(\mathbb{R}_+\times\mathbb{R}^n)\;|\;\gamma_0(\psi)\in D(\widetilde\Theta)\,,\quad\gamma_1(\psi)=\widetilde\Theta\gamma_0(\psi)\}\,,
\end{align*}
where, for smooth $\psi$, $(\gamma_0\psi)(y)=\psi(0,y)$ while $(\gamma_1\psi)(y)=\partial_x\psi(0,y)$.
Notice that, since $1\otimes\mathcal{F}_{\mathbb{R}^n}\colon L^2(\mathbb{R}_+\times\mathbb{R}^n)\to L^2(\mathbb{R}_+\times\mathbb{R}^n)$ is a unitary operator, each self adjoint extension $\widetilde{\Delta}_{\widetilde{\Theta}}$ corresponds to a self adjoint counterpart $\Delta_\Theta$ where $\widetilde{\Theta}\doteq(1\otimes\mathcal{F}_{\mathbb{R}^n})^{-1}\circ\Theta\circ(1\otimes\mathcal{F}_{\mathbb{R}^n})$.

In order to evaluate the spectrum of $\widetilde{\Delta}_{\widetilde{\Theta}}$, first we focus on the Dirichlet self-adjoint extension
\begin{align*}
D(\widetilde{\Delta}_\infty)\doteq\{\psi\in \widetilde{H}^2(\mathbb{R}_+\times\mathbb{R}^n)\;|\;\gamma_0(\psi)=0\}\,.
\end{align*}
For each $\psi\in D(\widetilde{\Delta}_\infty)$, $\textrm{ext}[\psi](x_1,k)=\theta(x_1)\psi(x_1,k)-\theta(-x_1)\psi(-x_1,k)$ identifies an element in $\widetilde{H}^2(\mathbb{R}^{n+1})=(1\otimes\mathcal{F}_{\mathbb{R}^{n+1}})H^2(\mathbb{R}^{n+1})$, $\theta$ being the Heaviside distribution.
The sine transform along the $x_1$-variable yields
\begin{align*}
\widetilde{\Delta}_\infty\psi(x_1,k)=\frac{2}{\pi}\int_{\mathbb{R}_+}\textrm{d}\xi\,(\xi^2+k^2)\sin(\xi x_1)
\bigg[\int_{\mathbb{R}_+}\textrm{d}y_1\,\sin(\xi y_1)\psi(y_1,k)\bigg]\,,
\end{align*}
from which we read that $\widetilde{\Delta}_\infty$ coincides with the multiplication operator by $\xi^2+k^2$, which entails, in turn, $\sigma(\widetilde{\Delta}_\infty)=\sigma(\Delta_{\infty})=(0,+\infty)$.
\\
The remaining contribution to the spectrum of $\widetilde{\Delta}_{\widetilde{\Theta}}$ can be studied via Theorem \ref{Theorem: spectral properties via Weil function}.
In particular the Weyl function associated with the boundary triple $(L^2(\mathbb{R}^n),\gamma_0,\gamma_1)$ is,
\begin{align*}
[M(\lambda)g](k)=-\sqrt{k^2-\lambda}\,g(k)\qquad
\textrm{ for }
\lambda\in\rho(\widetilde{\Delta}_\infty)\cap\mathbb{R}=(-\infty,0)\,.
\end{align*}
Still applying Theorem \ref{Theorem: spectral properties via Weil function} we obtain that $\lambda\in\sigma(\widetilde{\Delta}_{\widetilde{\Theta}})$ if and only if $0\in\sigma(\widetilde{\Theta}-M(\lambda))=\sigma(\widetilde{\Theta}+\sqrt{k^2-\lambda})$.
\\
At this stage we specialize to a specific scenario: Robin boundary conditions. In other words $\Theta=\alpha\mathbb{I}$, $\mathbb{I}$ being the identity operator on $L^2(\mathbb{R}^n)$, while $\alpha\in\mathbb{R}$. Hence also  $\widetilde{\Theta}=\alpha\mathbb{I}$ and it descends
$$\sigma(\widetilde{\Theta}+\sqrt{k^2-\lambda})=\{\alpha+\sqrt{k^2-\lambda}\;|\; k\in\mathbb{R}\}\,.$$
Thus if $\alpha\geq 0$, $0\notin\sigma(\Theta-M(\lambda))$, while, if $\alpha< 0$, all negative values of $\lambda$ greater that $-\alpha^2$ are allowed.
Hence $\widetilde{\Delta}_{\alpha\mathbb{I}}$ is bounded from below for all $\alpha\in\mathbb{R}$ with $\sigma(\widetilde{\Delta}_{\alpha\mathbb{I}})=(-\alpha^2,+\infty)$ for $\alpha< 0$, while $\sigma(\widetilde{\Delta}_{\alpha\mathbb{I}})=(0,+\infty)$ for $\alpha\geq 0$.
A special scenario occurs if we consider $n=0$, since the Weyl function reads $M(\lambda)=\sqrt{-\lambda}$. In this case, if $\alpha<0$, the only admissible, negative value for $\lambda$ is $-\alpha^2$ and thus all conditions of Theorem \ref{Theorem: spectral properties via Weil function} are met. We observe that this result is consistent with that of \cite{Dappiaggi:2017wvj} for the study of a massless, conformally coupled real, scalar field on the Poincar\'e patch of AdS spacetime.
\\
We can thus write the causal propagator $\widetilde{\mathsf{G}}_{\alpha\mathbb{I}}$ for $\partial_t^2+\widetilde{\Delta}$ explicitly.
This is tantamount to computing $\mathsf{G}_{\alpha\mathbb{I}}$ up to the Fourier transform in all spatial variables barring $x_1$.
We have
\begin{align*}
\big[\,\widetilde{\mathsf{G}}_{\alpha\mathbb{I}}f\,\big](t,x_1,k)&=
\int_{\mathbb{R}}\int_{\mathbb{R}_+}\sin\big[(\xi^2+k^2)^{\frac 12}(t-s)\big](\xi^2+k^2)^{-\frac 12}\psi(x_1,\xi)\bigg[\int_{\mathbb{R}_+}\psi(y_1,\xi)f(s,y_1,k)\textrm{d}y_1\bigg]\textrm{d}\xi\,\textrm{d}s
\\ &-
\int_{\mathbb{R}}\sinh\big[(\alpha^2-k^2)^{\frac 12}(t-s)\big](\alpha^2-k^2)^{-\frac 12} 1_{(-\alpha,\alpha)}(k)e_\alpha(x_1)\bigg[\int_{\mathbb{R}_+}e_\alpha(y_1)f(s,y_1,k)\textrm{d}y_1\bigg]\textrm{d}s\,,
\end{align*}
where we defined
$$\psi(x_1,\xi)\doteq(\xi^2+\alpha^2)^{-\frac 12}[\xi\cos(\xi x_1)+\alpha\sin(\xi x_1)]\,\qquad e_\alpha(x)\doteq (2\alpha)^{\frac 12}e^{-\alpha x}\,.$$
Another relevant choice for $\Theta$ is a differential operator.
This may model a Riemannian version of the well-known Wentzell boundary conditions (which will be treated in the next section).
In particular, in this latter case the operator $\widetilde{\Theta}$ acts as $[\widetilde{\Theta}f](k)=p_\Theta(k)f(k)$ where $p_\Theta$ is the symbol associated with $\Theta$ via Fourier transform.
Thus the condition for $\lambda<0$ to be in the spectrum of $\Delta_\Theta$ is that
$0\in\{p_\Theta(k)-\sqrt{k^2-\lambda}|\quad k \in\mathbb{R}\}\,.$
\end{Example}

\subsection{Dynamical boundary conditions}\label{Sec:Dynamical_boundary_condition}
Theorem \ref{Theorem: construction of the advanced and retarded propagators} and the analysis of the previous section were tied to the construction of the advanced and retarded Green's operators $\mathsf{G}^\pm_\Theta$ in terms of boundary conditions ascribed to a choice both of boundary triple and of a self-adjoint operator $\Theta$ on the boundary Hilbert space $L^2(\partial M)$. From the Lorentzian viewpoint, this scenario can be interpreted as assigning via $\Theta$ a time-independent boundary condition. 

Nonetheless, in many models and applications, this is not sufficient since one wishes to relax the hypothesis of time-independence, see {\it e.g.} \cite{Feller57, Zahn18} for a specific example. We show how it is possible to extend our previous analysis to encompass also most of these scenarios.

In the rest of the section, we consider $(N,h)$, a generic Lorentzian, static spacetime, with timelike boundary 
where $h$ is given by equation \eqref{Equation: conformally rescaled metric}.
We assume that the warping factor $\beta$ is chosen in such a way that, writing the D'Alembert wave operator as $\square = \partial^2_t+A$ with $A$ defined as per equation \eqref{Equation: elliptic operator with beta dependence}, $A$ is a uniformly elliptic operator on $L^2(M)$.

In addition, let $(\gamma_0,\gamma_1,L^2(\partial M))$ be the boundary triple associated to $A^*$ as per \eqref{eq:boundary_triple_for_A*}. With $(\widetilde{\gamma}_0,\widetilde{\gamma}_1)$ we denote the natural extension of $(\gamma_0,\gamma_1)$ to $C^\infty(\mathbb{R};D(A^*))$, while the boundary condition is implemented by restricting to the subspace of $C^\infty(\mathbb{R};D(A^*))$ whose elements $u$ satisfy

\begin{align}\label{Equation: dynamical boundary conditions}
\widetilde\gamma_1 u =(\partial_t^2+\widetilde\Theta)\widetilde\gamma_0 u\,,
\end{align}
where $\widetilde{\Theta}$ is the natural extension to $C^\infty(\mathbb{R},D(\Theta))$ of a self-adjoint operator $\Theta\colon D(\Theta)\subseteq L^2(\partial M)\to L^2(\partial M)$. Observe that, being the background static, $\partial_t^2\widetilde\gamma_0=\widetilde\gamma_0\partial_t^2$.

In order to address the problem of existence and uniqueness of the advanced and retarded Green's operators for the system at hand we make use of a technique which extends the one used in \cite{Feller57,Zahn18}. It consists of rewriting the wave equation \eqref{eq:PDE_for_G} with boundary conditions prescribed via \eqref{Equation: dynamical boundary conditions} in terms of the following equivalent system: Let $\widehat{u}=(u,v)\in C^\infty(\mathbb{R},D(A^*))\times C^\infty(\mathbb{R},D(\Theta))$
\begin{equation}\label{eq:auxiliary_system}
\widehat{\mathsf{P}}\widehat{u}\doteq\partial_t^2\widehat{u}+\widehat{A}_\Theta\widehat{u}=0\,,\quad
\widehat{A}_\Theta=\bigg(\begin{matrix}
\widetilde{A}^*&0\\-\widetilde\gamma_1&\widetilde\Theta
\end{matrix}
\bigg)\,,\qquad
\widetilde\gamma_0(u)=v\,,
\end{equation}
where $\widetilde{A}^*$ indicates the natural extension of $A^*$ to $C^\infty(\mathbb{R};D(A^*))$.

Observe that in \eqref{eq:auxiliary_system} we have introduced the auxiliary operator $\widehat{A}_\Theta$ which acts on $C^\infty(\mathbb{R};D(\widehat{A}_\Theta))$, where $D(\widehat{A}_\Theta)$ is the subset of
$\widehat{\mathsf{H}}\doteq\mathsf{H}\oplus\mathsf{h}=L^2(M)\oplus L^2(\partial  M)$, the Hilbert space with scalar product
\begin{gather*}
(\widehat{\phi}_1|\widehat{\phi}_2)_{\widehat{\mathsf{H}}}=
(\phi_1|\phi_2)_{\mathsf{H}}+(\varphi_1|\varphi_2)_{\mathsf{h}}\,, \quad\widehat{\phi}_i=(\phi_i,\varphi_i),\;i=1,2,
\end{gather*}
such that
\begin{align}\label{Equation: definition of the lifted S operator}
	D(\widehat{A}_\Theta)\doteq\{
	\widehat{\phi}=(\phi,\varphi)\in\widehat{\mathsf{H}}\;|\;
	\phi\in H^2( M)\,,
	\varphi\in D(\Theta)\,,
	\gamma_0\phi=\varphi
	\}\,,\qquad
	\widehat{A}_\Theta\widehat{\phi}=\bigg(
	\begin{array}{l}
	A^*\phi \\ \Theta\varphi-\gamma_1\phi
	\end{array}
	\bigg)\,.
\end{align}

The Hilbert space structure on $\widehat{\mathsf{H}}$ is defined by
\begin{align*}
	(\widehat{\phi}_1|\widehat{\phi}_2)_{\widehat{\mathsf{H}}}=
	(\phi_1|\phi_2)_{\mathsf{H}}+(\varphi_1|\varphi_2)_{\mathsf{h}}\,, \qquad\widehat{\phi}_i=(\phi_i,\varphi_i),\;i=1,2\,.
\end{align*}
We recall that $\Theta\colon D(\Theta)\subseteq \mathsf{h}\to\mathsf{h}$ is a self-adjoint operator on $\mathsf{h}$ and that $(\mathsf{h},\gamma_0,\gamma_1)$ is a boundary triple for $A^*$. We can prove a relevant property of $\widehat{A}_\Theta$.

\begin{proposition}\label{Proposition: self adjointness of lifted operator}
	Let $(\mathsf{h},\gamma_0,\gamma_1)$ be the boundary triple associated with $A^*$ as per Proposition \ref{Proposition: boundary triple of the Laplacian}. Let $\Theta\colon D(\Theta)\subseteq \mathsf{h}\to \mathsf{h}$ be a self adjoint operator.
	Then the operator $\widehat{A}_\Theta\colon D(\widehat{A}_\Theta)\subseteq\widehat{\mathsf{H}}\to\widehat{\mathsf{H}}$, defined as per equation \eqref{Equation: definition of the lifted S operator}, is self-adjoint.
\end{proposition}

\begin{proof}
	Since $\gamma_0$ is a surjective map, while $D(\Theta)$ is dense in $\mathsf{h}$ and $H^1_0(M)=\ker\gamma_0$, $D(\widehat{A}_\Theta)$ is dense in $\widehat{\mathsf{H}}$.
	On account of equation \eqref{Equation: boundary equation for boundary triples} it holds $(\widehat{A}_\Theta\widehat{\phi}|\widehat{\psi})=(\widehat{\phi}|\widehat{A}_\Theta\widehat{\psi})$ for all $\widehat{\phi},\widehat{\psi}\in D(\widehat{A}_\Theta)$, that is, $\widehat{A}_\Theta$ is a symmetric operator.
	We prove that $D(\widehat{A}_\Theta^*)\subseteq D(\widehat{A}_\Theta)$.
	Let $\widehat{\phi}\in D(\widehat{A}_\Theta^*)$.
	The map $D(\widehat{A}_\Theta)\ni\widehat{\rho}\mapsto(\widehat{\phi}|\widehat{A}_\Theta\widehat{\rho})\in\mathbb{C}$ is bounded.
	A direct inspection shows that, if $\widehat{\phi}=(\phi,\varphi)$ and $\widehat{\rho}=(\rho,\varrho)$,
	\begin{gather}\label{Equation: technical equation for proof}
	(\widehat{\phi}|\widehat{A}_\Theta\widehat{\rho})_{\widehat{\mathsf{H}}}=
	(\phi|A^*\rho)_{\mathsf{H}}+(\varphi|\Theta\varrho-\gamma_1\rho)_{\mathsf{h}}\,.
	\end{gather}
	If $\rho\in\ker\gamma_1\cap\ker\gamma_0=H^2_0(M)$, the last term on the right hand side of equation \eqref{Equation: technical equation for proof} vanishes and the boundedness of the left hand side implies that $\phi\in D(A^*)=H^2(M)$.
	By applying equation \eqref{Equation: boundary equation for boundary triples} we find
	\begin{align}\label{Equation: action of the adjoint lifted operator on the first component}
		(\widehat{A}_\Theta^*\widehat{\phi}|\widehat{\rho})_{\widehat{\mathsf{H}}}=
		(\widehat{\phi}|\widehat{A}_\Theta\widehat{\rho})_{\widehat{\mathsf{H}}}=
		(\phi|A\rho)_{\mathsf{H}}=(A^*\phi|\rho)_{\mathsf{H}}\,,
	\end{align}
	where in the last equality we used equation \eqref{Equation: boundary equation for boundary triples}.
	This shows that $\textrm{pr}_1\widehat{A}_\Theta^*\widehat{\phi}=A^*\phi$ where $\textrm{pr}_1\colon\widehat{\mathsf{H}}\to L^2( M)$ is the projection on the first component.
	
	Let us now consider $\widehat{\rho}\in D(\widehat{A}_\Theta)$.
	Since $\phi\in H^2(M)$, we can combine equations \eqref{Equation: technical equation for proof} and equation \eqref{Equation: boundary equation for boundary triples} for the boundary triple $(\mathsf{h},\gamma_0,\gamma_1)$ to obtain
	\begin{align}\label{Equation: simplified technical equation for proof}
	(\widehat{\phi}|\widehat{A}_\Theta\widehat{\rho})_{\widehat{\mathsf{H}}}=
	(A^*\phi|\rho)_{\mathsf{H}}
	+(\gamma_0\phi|\gamma_1\rho)_{\mathsf{h}}-(\gamma_1\phi|\varrho)_{\mathsf{h}}+(\varphi|\Theta\varrho)_{\mathsf{h}}-(\varphi|\gamma_1\rho)_{\mathsf{h}}\,.
	\end{align}
	In particular, the boundedness condition on $\widehat{\rho}\mapsto(\widehat{\phi}|\widehat{A}_\Theta\widehat{\rho})$ implies that $\varphi\in D(\Theta^*)=D(\Theta)$. Therefore equation \eqref{Equation: simplified technical equation for proof} simplifies further and we conclude that
	\begin{align}
	(\widehat{A}_\Theta^*\widehat{\phi}|\widehat{\rho})_{\widehat{\mathsf{H}}}-
	(\widehat{A}_\Theta\widehat{\phi}|\widehat{\rho})_{\widehat{\mathsf{H}}}=
	(\gamma_0\phi-\varphi|\gamma_1\rho)_{\mathsf{h}}\,.
	\end{align}
	If $\rho\in\ker\gamma_0$, the left hand side vanishes -- we recall that $\textrm{pr}_1\widehat{A}_\Theta^*\widehat{\phi}=A^*\phi$.
	Since $(\gamma_0,\gamma_1)$ is a surjective map, it follows that $\gamma_0(\phi)=\varphi$, that is, $\widehat{\phi}\in D(\widehat{A}_\Theta)$.
\end{proof}
\begin{remark}
Proposition \ref{Proposition: self adjointness of lifted operator} can be easily generalized to the case of an arbitrary boundary triple $(\mathsf{h},\gamma_0,\gamma_1)$ for the adjoint $S^*$ of a densely defined symmetric operator $S$ on an Hilbert space $\mathsf{H}$
\footnote{
		For the sake of completeness, one should assume that $\gamma_0$ is \textit{not} closable with respect to the norm topology of $\mathsf{H}$ -- in order to ensure that $D(\widehat{A}_\Theta)$ is dense in $\widehat{\mathsf{H}}$: This is the case in all applications.}.
In this latter case, $\Theta$ is a densely defined self-adjoint operator on $\mathsf{h}$, while the operator $\widehat{S}_\Theta$ is defined on $D(\widehat{S}_\Theta)\subset\widehat{\mathsf{H}}\doteq\mathsf{H}\oplus\mathsf{h}$, in particular,
$$D(\widehat{S}_\Theta)\doteq\{\widehat{\phi}\in\widehat{\mathsf{H}}|\;\phi\in D(S^*)\,,\;\varphi\in D(\Theta)\,,\;\gamma_0\phi=\varphi\}\,,$$
where $\widehat{\phi}=(\phi,\varphi)$.
\end{remark}
\begin{remark}
The spectral properties of $\widehat{A}_\Theta$ can be linked to elliptic boundary value problems with spectral-dependent boundary conditions, see for example \cite{Beh10} and \cite{BL12}.
Indeed, $\lambda\in\rho(\widehat{A}_\Theta)$ implies that $\lambda\in\rho(A_{\Theta-\lambda})$, where $A_{\Theta-\lambda}$ denotes the self-adjoint extension of $A$ defined via Proposition \ref{Proposition: characterization of self adjoint extension via boundary triples}.
\end{remark}

We can now formulate the counterpart of Theorem \ref{Theorem: construction of the advanced and retarded propagators} for the case at hand. Since the proof is the same mutatis mutandis, we will omit it:

\begin{theorem}\label{Theorem: advanced and retarded propagators for dynamical BC}
	Let $(\mathsf{h},\gamma_0,\gamma_1)$ be the boundary triple associated with $A^*$ as per Proposition \ref{Proposition: boundary triple of the Laplacian}.
	Let $\Theta\colon D(\Theta)\subseteq \mathsf{h}\to \mathsf{h}$ be a self-adjoint operator and let assume that the spectrum of the self-adjoint operator $\widehat{A}_\Theta$ defined in Proposition \ref{Proposition: self adjointness of lifted operator} is bounded from below.
	Finally let $\textrm{pr}_1\colon\mathsf{H}\to L^2(M)$ be the projection to the first component of $\widehat{\mathsf{H}}$ and let $\textrm{ext}\colon H^2( M)\to\widehat{\mathsf{H}}_{\gamma_0}$ be the extension map $\phi\mapsto\widehat{\phi}_{\textrm{ext}}\doteq(\phi,\gamma_0\phi)$.
	For $f\in C^\infty(\mathring{N})$, let
	\begin{align}\label{Equation: advanced and retarded propagators for dynamical BC}
		[\mathsf{G}_\Theta f](t)\doteq
		\textrm{pr}_1\bigg[\int_{\mathbb{R}}\widehat{A}_\Theta^{-\frac 12}\sin\big[\widehat{A}_\Theta^{\frac 12}(t-t^\prime)\big]\widehat{f}_{\textrm{ext}}(t^\prime)\textrm{d}t^\prime\bigg]\,.
	\end{align}
	Then $(\mathsf{G}^+_\Theta,\mathsf{G}^-_\Theta)$ where $\mathsf{G}^+=\theta(t-t^\prime)\mathsf{G}$ and $\mathsf{G}^-=-\theta(t^\prime-t)\mathsf{G}$ is a pair of advanced and retarded Green's operators for $\mathsf{P}$ as per Definition \ref{Definition: advanced, retarded and causal propagator} subjected to the boundary conditions \eqref{Equation: dynamical boundary conditions}.
\end{theorem}

\begin{Example}
With reference to the same geometric setting as in Example \ref{Example: propagators for half space} we apply the previous analysis, Theorem \ref{Theorem: advanced and retarded propagators for dynamical BC} in particular, to the dynamical boundary condition \eqref{Equation: dynamical boundary conditions} defined out of the  operator  $\Theta=-\sum_{\ell=2}^{n+1}\partial_{x_\ell}^2$.
Following \cite{Zahn18}, a normalized complete system of generalized eigenfunctions is
\begin{align}
\widehat{\psi}_{\xi,k}(x_1,x)=\big[a(\xi^2+1)\big]^{-\frac 12}\bigg(\begin{array}{l}
e^{ikx}\big(\cos(\xi x_1)-\xi\sin(\xi x_1)\big)\\
e^{ikx}
\end{array}\bigg)\,,
\end{align}
where $a=(2\pi)^{n-1}\frac\pi 2$. A direct computation yields
\begin{align*}
\big[\textrm{pr}_1\widehat{\Delta}\widehat{\phi}\big](x_1,x)=&
\int_{\mathbb{R}_+\times\mathbb{R}}(k^2+\xi^2)\psi_{\xi,k}(x_1,x)\bigg[\int_{\mathbb{R}_+\times\mathbb{R}}\psi_{\xi,k}(y_1,y)\phi(y_1,y)\textrm{d}y_1\textrm{d}y\bigg]\textrm{d}\xi\textrm{d}k\\
[\mathsf{G}_\Theta f](t,x_1,x)=&
\int_{\mathbb{R}}\int_{\mathbb{R}_+\times\mathbb{R}}\!\!\!\!\sin\big[(\xi^2+k^2)^{\frac 12}(t-s)\big](\xi^2+k^2)^{-\frac 12}\psi_{\xi,k}(x_1,x)\cdot\\
&\cdot\bigg[\int_{\mathbb{R}_+\times\mathbb{R}}\!\!\!\psi_{\xi,k}(y_1,y)f(s,y_1,y)\textrm{d}y_1\textrm{d}y\bigg]\textrm{d}\xi\textrm{d}k\textrm{d}s\,.
\end{align*}
\end{Example}

\section*{Acknowledgements}
The work of C.~D.\ was supported by the University of Pavia.
The work of N.~D.\ was supported in part by a research fellowship of the University of Pavia. The work of N.~D.\ and of H.~F.\ was supported in part by a fellowship of the ``Progetto Giovani GNFM 2017'' under the project ``Wave propagation on Lorentzian manifolds with boundaries and applications to algebraic QFT" fostered by the National Group of Mathematical Physics (GNFM-INdAM).
We are grateful to Felix Finster, Nadine Grosse, Valter Moretti, Simone Murro and Juan Manuel P\'erez-Pardo for the useful comments and discussions. We are grateful to Igor Khavkine for the useful comments, especially concerning Proposition \ref{Prop:exact_sequence}.

\end{document}